%% file: paper.tex
\documentclass[11pt, a4paper]{article}
\input{preamble}
\graphicspath{{../}}

\begin{document}
\clearpage
\pagestyle{plain}
\setcounter{page}{1}
\include{title}

\include{main}

\bibliographystyle{agsm}
\bibliography{bibliography}

\clearpage

\appendix
\input{appendix}

\clearpage
\input{supplementary_material}

\end{document}

%% file: preamble.tex
\usepackage{natbib}
\setcitestyle{authoryear,open={(},close={)}}
\usepackage[utf8]{inputenc}
\usepackage[T1]{fontenc}
\usepackage{lmodern,textcomp}
\usepackage{titlesec}
\usepackage{graphicx}
\usepackage{epstopdf}
\usepackage{caption}
\usepackage{subcaption}
\usepackage{float}
\usepackage[space]{grffile}
\usepackage{xr-hyper}
\usepackage{hyperref}
\hypersetup{colorlinks=true,citecolor=blue}
\usepackage{rotating}
\usepackage{pdflscape}
\usepackage{amsmath}
\usepackage{acronym}
\usepackage{booktabs}
\usepackage{siunitx}
\usepackage[UKenglish]{babel}
\usepackage{microtype}
\usepackage{multirow}
\usepackage{tabularx}
\usepackage[retainorgcmds]{IEEEtrantools}
\usepackage{bm}
\usepackage{color}
\usepackage{longtable}
\usepackage{array}
\usepackage[normalem]{ulem}
\usepackage{authblk}
\usepackage{geometry}
\usepackage{makecell}
\usepackage{amsthm}
\usepackage{amssymb}
\usepackage{eurosym}
\usepackage{tfrupee}
\usepackage{breakcites}

\newtheorem{remark}{Remark}
\newtheorem{proposition}{Proposition}
\newtheorem{appxprop}{Proposition}[section]

%% file: title.tex
\setcounter{equation}{0}

\title{Revisiting Strategic Protectionism in Network Markets}
\author[a]{Nigar Hashimzade}
\author[b]{Haoran Sun}
\affil[a]{Department of Economics, Finance and Accounting, Brunel University of London, Kingston Lane, Uxbridge, UB8 3PH, United Kingdom\\
E-mail: nigar.hashimzade@brunel.ac.uk}
\affil[b]{Corresponding author\\
Department of Economics, Durham University, Mill Hill Lane, Durham, DH1 3LB, United Kingdom\\
E-mail: haoran.sun@durham.ac.uk}
\date{This version: \today}

\renewcommand\Affilfont{\small}


\maketitle

\begin{abstract}
\normalsize
This paper studies the welfare effects of an industrial policy targeting a sector with network externalities in a two-country model with strategic trade and R\&D investment. When externalities are weak or the goods are close substitutes, the business-stealing effect dissipates more surplus than the externality creates. Under sufficiently strong externalities and weak substitutability or complementarity of the goods, industrial policy competition can make both countries simultaneously better off compared to the laissez-faire outcome because of the mutual business-enhancement effect. The case is stronger for product innovation than for process innovation, as the former directly affects the demand and triggers a stronger network effect than the latter which operates indirectly through the industry supply. Thus, the network externalities create an opportunity for win-win industrial policies, but its realisation depends on the market structure and the nature of innovation.

\noindent \textbf{JEL Classification:} F13, H25, L13, O38

\bigskip

\noindent \textbf{Keywords:} Industrial policy, network externalities, R\&D subsidies, strategic trade policies, Cournot competition

\end{abstract}

%% file: main.tex
\setcounter{equation}{0}

\section{Introduction}\label{sec:intro}

Industrial policy has returned to the centre of economic policymaking. In the semiconductor sector, the US CHIPS and Science Act is a \$52.7 billion commitment to semiconductor manufacturing and research within the US \citep{ChipsAct_2022}, while the European Union (EU) Chips Act intends to raise over \euro{}43 billion \citep{EuChipsAct_2023}. China recently announced phase three of its National Integrated Circuit Industry Investment Fund supported with \yen 344 billion (about \$47.5 billion at May 2024 exchange rates) of state capital \citep{ChinaBigFund_2024}. Similarly, India's Semicon India programme received \rupee{}76,000 crore (about \$10 billion) from the Indian government \citep{MeitySemicon_2021}. Support for electric vehicles and clean technology supply chains is also substantial across major economies. For example, the US Inflation Reduction Act established advanced manufacturing tax credits (Section 45X) reserved for EV battery cells, critical minerals, and renewable energy components manufactured in the United States \citep{Ira_2022}. The One Big Beautiful Bill Act of 2025 subsequently restricted foreign entities of concern (FEOCs) from receiving those tax credits \citep{Obbba_2025}. The European Union's Net-Zero Industry Act establishes manufacturing capacity targets and accelerates regulatory approval processes for domestic manufacturers of clean technologies \citep{EuNzia_2024}. China subsidised domestic production of new energy vehicles until the end of 2022 and imposed a tradeable credit obligation on all passenger vehicle producers \citep{MiitNev_2020,MofNev_2020}. India introduced Production Linked Incentive (PLI) programmes to provide support to both battery and automotive component suppliers totalling over \rupee{}44,000 crore (about \$6 billion at 2021 exchange rates) \citep{MhiPliAuto_2021, MhiPliAcc_2021}. These programmes concentrate on sectors in which the value of a product rises with the breadth of its adoption, through compatibility, complementary services, and shared infrastructure. In other words, these sectors are characterised by strong network externalities in demand. When governments target such industries they use subsidies to encourage domestic companies and import tariffs, digital services taxes and restrictive regulations to discourage their foreign competitors. Existing work on industrial policies in a trade context does not take into consideration network externalities in demand. Thus, whether this combination of policies will increase aggregate welfare or be counterproductive remains an open question. We develop such a framework and characterise the circumstances where the industrial and trade policy mix raises or destroys welfare.

The classical economics answer to whether a government should support its home industry against foreign competition is ``No''. Government intervention often causes distortion in allocation that comparative advantage produces, creates the opportunity for other countries to retaliate, and ultimately destroys gains from trade \citep{Krueger_1990, Pack.Saggi_2006}. A more recent argument suggests that some level of intervention may be acceptable on grounds of supply chain resilience, technological sovereignty, or the political economy of employment, among other factors \citep{Millot.Rawdanowicz_2024, Juhasz.etal_2024}. Although most of the current literature accepts these motives, it still treats government intervention as a way of achieving a strategic or distributional objective at the price of a deadweight loss caused by distortionary policies.

This paper explores an alternative argument. When demand exhibits positive network externalities, market equilibrium consumption is below the efficient level. Each buyer weighs her own benefit against the price and ignores the value her purchase adds for other buyers. Therefore, rather than distorting efficient allocation, industrial policies intended to increase production in sectors exhibiting network externalities correct an inefficient one. A well-designed policy will reduce or even eliminate deadweight loss. Furthermore, when countries trade in goods with  consumption externalities, the importing country can optimally choose to subsidise imports. Non-cooperative policy competition can then leave both importing and exporting countries better off than laissez-faire.

We develop this analysis in a two-country, two-firm model in which the firms sell differentiated goods with positive network externalities in demand, invest in R\&D, and compete in quantities. Innovation takes one of two forms, cost-reducing process R\&D or demand-enhancing product R\&D. Each government subsidises its own firm's R\&D expenses. Additionally, the government of the home country, where consumers purchase home and foreign goods, imposes a tax on the foreign firm's revenue from sales in the home market. We compare the resulting policy equilibrium with the laissez-faire outcome.

The model generates several predictions that have direct implications for the assessment of observed industrial policies. First, when the demands for the two goods are interdependent, both governments subsidise R\&D in equilibrium, and these subsidies are strategic complements. That is, supporting R\&D activity unilaterally creates an incentive for the foreign government to do the same. This is consistent with the sequence of subsidy commitments in semiconductors and electric vehicles across the US, the EU, and China. Second, the sign of the optimal tax instrument depends on the strength of consumption externalities and the extent of product differentiation. In the interior equilibrium, the optimal tax tends to decline as the network externality strengthens, and can become negative (turn into an import subsidy) once the externality is sufficiently strong and the goods are weak substitutes or complements. In this case, as the home government acts to correct the externality in consumption, the increase in the consumer welfare from higher consumption of both the home and the foreign goods outweighs the fiscal cost of lowering import tax or even subsidising the foreign firm. This effect is more pronounced for the product R\&D than for the process R\&D, because the former increases the consumers' willingness-to-pay directly, whereas the latter works indirectly, through the lower production cost and, hence, lower prices.

The aggregate welfare effect of policy competition depends on the strength of the externality. When network externalities are weak, intervention is damaging on aggregate. The distortion from taxes and subsidies outweighs the gain from correcting the externality. At intermediate strength the policy competition redistributes surplus, with the home country gaining more than the foreign country loses. When the externality is sufficiently strong and the goods are sufficiently differentiated (weak substitutes or complements), policy competition makes both countries simultaneously better off compared to the laissez-faire outcome. We characterise the win-win region across two parameters, the degree of product differentiation and the strength of network externality. The numerical exercise\footnote{In all numerical examples we focus on the range of parameters leading to interior equilibria. We refer to it as the \emph{admissible range}.} shows that the win-win configuration of these two parameters covers the entire admissible range of complementarity and a non-trivial subset of substitutability. The minimal network externality strength required for the win-win outcome decreases as the goods become more complementary. Under process R\&D the win-win set is smaller than under product R\&D: for the same degree of externality the weaker indirect effect of process R\&D requires weaker substitutability (or stronger complementarity) between the goods for the foreign welfare to rise.

The theoretical framework allows an alternative interpretation, also relevant for industrial policies. Consider a downstream industry located in the home country that uses domestically produced and imported intermediate inputs in production. The production function then replaces the consumer utility in the model, and the derived demand for inputs replaces the consumer demand. Consumer surplus is equivalent to the downstream industrial rent. The network externality in the demand for inputs represents compatibility, shared technical standards, and the pool of specialised complementary services. Since all welfare results continue to hold, the analysis applies to competition for dominance in intermediate-good markets as much as to consumer goods markets. 

Our paper bridges the literature on strategic trade policy with the economics of network externalities to provide a formal evaluation of modern industrial policy. Strategic trade policy in international oligopoly was first analysed in seminal papers by \cite{Spencer.Brander_1983}, \cite{Brander.Spencer_1985} and  \cite{Eaton.Grossman_1986}. These studies showed that government intervention can shift oligopoly rents and that optimal policy depends on the interaction between market structure and the mode of competition. Since direct export subsidies are prohibited under the WTO Agreement on Subsidies and Countervailing Measures, a prominent strand of this literature focuses on subsidies to R\&D and innovation.

This line of research has studied a wide range of non-cooperative and cooperative R\&D policy options \citep{Leahy.Neary_2001, Haaland.Kind_2006, Haaland.Kind_2008}, sequential choice and complementarity between process and product innovation \citep{Lin.Saggi_2002}, product innovation in international markets \citep{Symeonidis_2003}, and first-mover advantages \citep{GarciaPires_2009, Long.etal_2011, Baik.Kim_2020}. However, these studies largely assume standard demand functions without network externalities.

At the same time, a separate literature has developed on the economics of network externalities \citep{Katz.Shapiro_1985, Katz.Shapiro_1986, Farrell.Saloner_1985, Farrell.Saloner_1986, Economides_1996}, which, among other issues, investigated the effect of the product or technology adoption externalities on strategic choice of trade policies. In an early contribution, \cite{Krishna_1988} showed that demand-side network externalities and consumer expectations can alter or even reverse standard optimal trade policy interventions. A related public-finance strand analysed the incidence of taxation in multi-sided and digital media markets \citep{Kind.etal_2008, Kind.etal_2009, Kind.etal_2010, Kind.Koethenbuerger_2018} and digital services taxes \citep{Cui.Hashimzade_2019, Kind.Schjelderup_2025}. More recently, \cite{Ghosh.etal_2024} showed that network externalities can reverse the sign of optimal trade instruments. Optimal choices of trade policy, including no-intervention choice, under network externalities, were also analysed by \cite{Lim_2026}. Yet, some of these studies model firm productive capabilities as fixed and do not consider the strategic choice of innovation with the possibility of government subsidies, while those which incorporate network externalities into R\&D competition focus on closed-economy settings \citep{Naskar.Pal_2020}.

To our knowledge, no existing framework combines strategic trade policy, government subsidies to innovation, and demand-side network externalities. Incorporating innovation subsidies into a  framework with externalities in consumption is particularly important for contemporary policy analysis because modern high-technology sectors, such as semiconductors, electric vehicles, and battery supply chains, are characterised by strong adoption spillovers, compatibility effects, and shared infrastructure \citep{Li.etal_2017, Long.Cai_2023}. The recent wave of industrial policies \citep{Aiginger.Rodrik_2020, Juhasz.etal_2024, Erten.etal_2025} has sparked renewed debate over whether government interventions represent unjustified protectionism or welfare-enhancing correctives. Modelling innovation subsidies alongside network externalities allows us to clarify when industrial policy is justified. We show when intervention fosters adoption externalities to yield win-win outcomes, and when it benefits the home country at the expense of the foreign country.

The remainder of the paper is structured as follows. Section~\ref{sec:model} presents the model. Section~\ref{sec:equilibrium} derives the equilibrium and describes its comparative static properties. Section~\ref{sec:simulations} contains the welfare analysis of policy competition, including the map of the win-win region. Section~\ref{sec:conclusion} concludes. The Appendix (available online) contains analytical derivations and formal statements of the comparative static results. Supplementary material (available online)\footnote{The Appendix and the supplementary material are included in the working paper version of this paper, available at \url{https://arxiv.org/abs/2603.29542}.} provides institutional background on the tax instrument and some additional numerical results.

\section{The Model}\label{sec:model}

Consider two countries, home and foreign, each with one firm producing a differentiated good that exhibits network externalities in demand. Good~1 is produced by the foreign firm (Firm~1) and good~2 by the home firm (Firm~2). There is also a num\'{e}raire  good, good~0, produced by a perfectly competitive sector using one unit of labour per unit of output, and traded freely and costlessly.

For tractability, we model consumers (or, equivalently, a downstream industry; see Remark~\ref{rem:isomorphism} below) only in the home country. Thus, the home country is the importer and the foreign country is the exporter. A representative consumer supplies inelastically $L$ units of labour, owns the home firm, and receives a lump-sum transfer from the government. Goods 1 and 2 exhibit network externalities in consumption and may be either substitutes or complements.

Firms engage in two distinct forms of research and development (R\&D): process innovation, which lowers unit marginal costs, and product innovation, which enhances product features and expands baseline demand. While process innovation directly improves production efficiency, product innovation acts as a demand-enhancing investment that increases consumers' willingness-to-pay. 

Each government subsidises its own firm's R\&D costs, and the home government taxes the foreign firm's sales revenue. The interaction is modelled as a two-stage simultaneous-move game:
\begin{description}
    \item[Stage~1.] The home and the foreign governments simultaneously announce their policies.
    \item[Stage~2.] The home and the foreign firms simultaneously choose output quantities and R\&D investments.
\end{description}
The equilibrium is a subgame-perfect Nash equilibrium (SPNE), obtained by backward induction.

\subsection{Demand}\label{sec:demand}

The representative home-country consumer maximises the quasi-linear utility function \citep{Singh.Vives_1984}:
\begin{equation}\label{eq:utility}
    \max_{q_0,\,q_1,\,q_2}\; U = A_1\, q_1 + A_2\, q_2 - \frac{1}{2}\!\left(q_1^2 + q_2^2 + 2m\,q_1 q_2\right) + q_0
\end{equation}
subject to the budget constraint
\begin{equation}\label{eq:budget}
    p_1 q_1 + p_2 q_2 + q_0 = L + \pi_2 + T,
\end{equation}
where $A_i>0, ~i=1,2$, $L$ is labour income, $\pi_2$ is the home firm's profit, and $T$ is a lump-sum government transfer, taken by the consumer as given. The parameter $m \in [-1,1]$ measures the degree of substitutability between goods~1 and~2. The goods are homogeneous when $m=1$, substitutes when $m > 0$, complements when $m < 0$, and independent when $m = 0$.

Maximisation yields the inverse demand functions:
\begin{equation}\label{eq:invdemand}
    p_i = A_i - q_i - m\,q_j, \qquad i \neq j,\quad i,j \in \{1,2\},
\end{equation}
where $A_i$ is the maximal willingness-to-pay for good $i$ when $q_j$ is zero. Two channels determine $A_i$. The first is the effect of network externality, measured by $a_i$, defined in Eq.~\eqref{eq:network}. The second is the product quality effect, measured by $r_i$. Firm~$i$ can increase $r_i$, and, hence, the consumers' willingness-to-pay for its good, by investing in product R\&D. We assume that these two effects are independent and additive,
\begin{equation}\label{eq:Ai}
    A_i = a_i + r_i.
\end{equation}

\paragraph{Network externalities.}
Following \cite{Leibenstein_1950} and \cite{Krishna_1988}, 
we model positive network externalities as a positive effect of the perceived aggregate consumption on the willingness-to-pay. For simplicity we assume that the externality operates only through consumption of the same good. Specifically:
\begin{equation}\label{eq:network}
    a_i = a + b_i Q_i,
\end{equation}
where $Q_i$ is the perceived aggregate consumption of good~$i$, $a > 0$ is a baseline demand parameter, and $b_i \geq 0$ measures the strength of the network externality for good~$i$. While there are no direct cross-product externalities, the perceived aggregate consumption of one good creates an indirect effect on the demand for the other good whenever $m \neq 0$, with a sign of the effect that depends on whether the goods are substitutes or complements (Remark~\ref{rem:indirect}).
 
In a fulfilled-expectations equilibrium, perceived and actual aggregate consumption coincide, $Q_i = q_i$.

Upon substitution of Eqs.~\eqref{eq:Ai} and~\eqref{eq:network} into Eq.~\eqref{eq:invdemand}, with $Q_i = q_i$, the demand system becomes:
\begin{align}
    p_1 &= a + r_1 - (1 - b_1)\,q_1 - m\,q_2, \label{eq:demand1}\\
    p_2 &= a + r_2 - (1 - b_2)\,q_2 - m\,q_1. \label{eq:demand2}
\end{align}
In what follows, we will assume $b_i \in [0,1)$, to focus on the equilibria with downward-sloping demands. 
\subsection{Firms}\label{sec:firms}

Each firm maximises its net profits. The foreign firm faces a tax at rate $t < 1$ on its sales revenue in the home market. Negative values of $t$ are interpreted as an import subsidy.%
\footnote{The analysis of the ad valorem tariff case is straightforward with reparameterisation. The revenue-based structure of the positive tax captures a class of policy instruments that have proliferated internationally since the mid-2010s, including digital services taxes, equalisation levies, and withholding taxes on cross-border service payments (see Supplementary Material, Section~\ref{sm:tax_instruments}). Subsidies that accrue to foreign producers are also documented in network-goods sectors. For example, Norway exempts electric vehicles from purchase taxes and imports nearly all of them. Germany's Umweltbonus was paid on foreign-made electric vehicles, and China's consumer subsidies for new energy vehicles covered Tesla's locally sold output.} The home firm faces no tax on its domestic sales.

Each firm $i$ can invest in two types of R\&D:
\begin{itemize}
    \item \textit{Process R\&D} reduces marginal production cost from $c_i$ to $(c_i - k_i)$, at a cost of $\frac{\varphi_i k_i^2}{2}$;
    \item \textit{Product R\&D} increases demand by $r_i$, at a cost of $\frac{\theta_i r_i^2}{2}$,
\end{itemize}
where $\varphi_i$ and $\theta_i$ are the R\&D cost parameters for process and product R\&D respectively. We require $c_i - k_i > 0$ in equilibrium. That is, marginal cost remains positive even after process innovation reduces it. Each government subsidises its own firm's R\&D at rates $s_i$ (process) and $\sigma_i$ (product), so that the firm's net R\&D cost is $(1 - s_i)\frac{\varphi_i k_i^2}{2} + (1 - \sigma_i)\frac{\theta_i r_i^2}{2}$.

The net profits are hence given by
\begin{align}
    \pi_1 &= (1-t)\,p_1\,q_1 - (c_1 - k_1)\,q_1 - (1-s_1)\,\frac{\varphi_1 k_1^2}{2} - (1-\sigma_1)\,\frac{\theta_1 r_1^2}{2}, \label{eq:profit1}\\[6pt]
    \pi_2 &= p_2\,q_2 - (c_2 - k_2)\,q_2 - (1-s_2)\,\frac{\varphi_2 k_2^2}{2} - (1-\sigma_2)\,\frac{\theta_2 r_2^2}{2}. \label{eq:profit2}
\end{align}

\subsection{Governments}\label{sec:governments}

The foreign government maximises its welfare, defined as the foreign firm's profit net of subsidy expenditure:
\begin{equation}\label{eq:W1}
    W_1 = \pi_1 - s_1\,\frac{\varphi_1 k_1^2}{2} - \sigma_1\,\frac{\theta_1 r_1^2}{2}.
\end{equation}

The home government maximises home welfare, defined as the indirect utility of the representative consumer: 
\begin{equation}\label{eq:W2}
    W_2 = CS + L + \pi_2 + T,
\end{equation}
where 
\begin{equation}\label{eq:CS}
    CS \equiv \frac{1}{2}\left(q_1^2 + q_2^2\right) + m\,q_1 q_2,
\end{equation}
describes consumer surplus following from Eqs.~\eqref{eq:utility} and~\eqref{eq:invdemand}. The home government balances its budget by returning to the consumers a lump-sum transfer, which is the difference between the import tax revenues and the R\&D subsidies:
\begin{equation}\label{eq:govbudget}
    T = t\,p_1\,q_1 - s_2\,\frac{\varphi_2 k_2^2}{2} - \sigma_2\,\frac{\theta_2 r_2^2}{2}.
\end{equation}
The transfer is positive when tax revenue exceeds the subsidies paid, and negative otherwise, in which case the consumers fund the difference. Since $L$ is a constant, it does not affect the optimisation problem and the welfare comparisons.

\begin{remark}[Supply-Chain Isomorphism]\label{rem:isomorphism}
\normalfont
The demand system~\eqref{eq:demand1}--\eqref{eq:demand2} is derived from a representative consumer's utility. It is also mathematically equivalent to the derived demand of a downstream industry that uses the two goods as intermediate inputs. Suppose a competitive downstream sector uses the two goods as inputs with the quadratic technology $F(q_1, q_2) = A_1 q_1 + A_2 q_2 - \tfrac{1}{2}\left(q_1^2 + q_2^2 + 2m\,q_1 q_2\right)$, with output measured in units of the num\'{e}raire, where $A_i$ rises with the industry-wide use of input $i$ as in Eqs.~\eqref{eq:Ai} and~\eqref{eq:network}. Profit maximisation gives $p_i = \partial F/\partial q_i$, which is exactly the inverse demand \eqref{eq:invdemand}, and the sector's rent $F - p_1 q_1 - p_2 q_2$ equals the consumer surplus in Eq.~\eqref{eq:CS}. In this interpretation the network externality in the usage of an input captures productive spillovers which raise the contribution of that input in the final good production as its adoption grows. The consumer surplus in the home government's objective function is therefore the home downstream industry's rent. The strategic equilibrium between the upstream producers does not depend on which interpretation is adopted, so the welfare results extend to competition for dominance in intermediate-good markets.
\end{remark}

\begin{remark}[Direct vs Indirect Cross-Product Externality]\label{rem:indirect}
    \normalfont
    Equation~\eqref{eq:network} describes a direct externality of the aggregate consumption of one good on the representative consumer's demand for the same good. When the two goods are substitutes or complements, this direct externality also creates an indirect effect of the aggregate consumption of each good on the demand for the other good. 
    One could also extend the framework to a direct cross-product externality by modelling the willingness to pay for one good as a function of the aggregate consumptions of both goods.
\end{remark}

\subsection{Profit Maximisation in Stage~2}\label{sec:stage2}

We solve the game by backward induction, beginning with Stage~2. Taking government policies $(t, s_1, \sigma_1, s_2, \sigma_2)$ as given, each firm simultaneously chooses output and R\&D investments to maximise profit. Throughout the paper, we will focus on the interior equilibrium solutions. We will refer to the set of parameters for which an interior equilibrium exists as the admissible set, or admissible parameters.\footnote{Our numerical exercise confirms that the admissible set has positive measure, i.e. is non-trivial.}

Substituting the demand functions~\eqref{eq:demand1}--\eqref{eq:demand2} into the profit functions~\eqref{eq:profit1}--\eqref{eq:profit2} and differentiating gives the first-order conditions for the foreign firm:
\begin{align}
    \frac{\partial \pi_1}{\partial q_1} &= (1-t)\left[a + r_1 - 2(1-b_1)\,q_1 - m\,q_2\right] - (c_1 - k_1) = 0, \label{eq:foc_q1}\\[4pt]
    \frac{\partial \pi_1}{\partial k_1} &= q_1 - (1-s_1)\,\varphi_1\,k_1 = 0, \label{eq:foc_k1}\\[4pt]
    \frac{\partial \pi_1}{\partial r_1} &= (1-t)\,q_1 - (1-\sigma_1)\,\theta_1\,r_1 = 0. \label{eq:foc_r1}
\end{align}

These yield:
\begin{equation}\label{eq:firm1_reaction}
    q_1 = \frac{1}{2(1-b_1)}\left[a + r_1 - m\,q_2 - \frac{c_1 - k_1}{1-t}\right],\quad
    k_1 = \frac{q_1}{(1-s_1)\,\varphi_1},\quad
    r_1 = \frac{(1-t)\,q_1}{(1-\sigma_1)\,\theta_1}.
\end{equation}

Similarly, the home firm's optimal choices satisfy:
\begin{equation}\label{eq:firm2_reaction}
    q_2 = \frac{1}{2(1-b_2)}\left[a + r_2 - m\,q_1 - (c_2 - k_2)\right],\quad
    k_2 = \frac{q_2}{(1-s_2)\,\varphi_2},\quad
    r_2 = \frac{q_2}{(1-\sigma_2)\,\theta_2}.
\end{equation}

In choosing R\&D, each firm equates the marginal benefit of investment to its marginal cost net of the subsidy. For process R\&D the marginal benefit is the cost saving on every unit produced, so it equals the firm's output. For product R\&D it is the additional revenue that a unit of quality generates on every unit sold, net of the import tax for the foreign firm. The reaction functions~\eqref{eq:firm1_reaction} and~\eqref{eq:firm2_reaction} therefore imply that each firm's R\&D investment with quadratic cost is proportional to its output. A higher subsidy or a lower R\&D cost parameter reduces the marginal cost of investment. Thus the firm invests more per unit of output before marginal cost equals marginal benefit.

\section{Equilibrium Analysis}\label{sec:equilibrium}

This section presents the SPNE. We derive the Stage~2 equilibrium of the general case in which both process and product R\&D are active (Section~\ref{sec:general_equilibrium}). After that, for tractability, we analyse each type of R\&D separately: process R\&D in Section~\ref{sec:processRD} and product R\&D in Section~\ref{sec:productRD}. 

\subsection{Equilibrium with Both R\&D Channels}\label{sec:general_equilibrium}

Eliminating $k_i$ and $r_i$ from Eqs.~\eqref{eq:firm1_reaction}--\eqref{eq:firm2_reaction} yields a linear system in $(q_1, q_2)$. For the foreign firm:
\begin{equation}\label{eq:general_q1}
    \Gamma_1\, q_1 = a - \frac{c_1}{1-t} - m\,q_2,
\end{equation}
where
\begin{equation}\label{eq:Gamma1}
    \Gamma_1 \;\equiv\; 2(1-b_1) - \frac{1-t}{(1-\sigma_1)\theta_1} - \frac{1}{(1-t)(1-s_1)\varphi_1}.
\end{equation}

For the home firm:
\begin{equation}\label{eq:general_q2}
    \Gamma_2\, q_2 = a - c_2 - m\,q_1,
\end{equation}
where
\begin{equation}\label{eq:Gamma2}
    \Gamma_2 \;\equiv\; 2(1-b_2) - \frac{1}{(1-\sigma_2)\theta_2} - \frac{1}{(1-s_2)\varphi_2}.
\end{equation}

In an interior equilibrium, solving the two best responses~\eqref{eq:general_q1} and~\eqref{eq:general_q2} for the quantities gives:
\begin{align}
    q_1 &= \frac{\displaystyle\frac{1}{\Gamma_1}\!\left[a - \frac{c_1}{1-t}\right] - \frac{m}{\Gamma_1\,\Gamma_2}\left[a - c_2\right]}{\Delta^{kr}}, \label{eq:q1_general}\\[8pt]
    q_2 &= \frac{\displaystyle\frac{1}{\Gamma_2}\left[a - c_2\right] - \frac{m}{\Gamma_1\,\Gamma_2}\!\left[a - \frac{c_1}{1-t}\right]}{\Delta^{kr}}, \label{eq:q2_general}
\end{align}
where
\begin{equation}\label{eq:Delta_general}
    \Delta^{kr} \;\equiv\; 1 - \frac{m^2}{\Gamma_1\,\Gamma_2}.
\end{equation}

Once $q_1$ and $q_2$ are determined, all four R\&D investments follow immediately from Eqs.~\eqref{eq:firm1_reaction} and~\eqref{eq:firm2_reaction}.

An interior equilibrium requires $\Gamma_i > 0$ and $\Delta^{kr} > 0$. 

\subsection{Process R\&D}\label{sec:processRD}

Setting $r_1 = r_2 = 0$ leads to the equilibrium with process R\&D described by
\begin{align}
    q_1 &= \frac{1}{2(1-b_1)}\left[a - mq_2 - \frac{c_1 - k_1}{1-t}\right], &
    k_1 &= \frac{q_1}{(1-s_1)\varphi_1}, \label{eq:proc_sys1}\\[4pt]
    q_2 &= \frac{1}{2(1-b_2)}\left[a - mq_1 - (c_2 - k_2)\right], &
    k_2 &= \frac{q_2}{(1-s_2)\varphi_2}. \label{eq:proc_sys2}
\end{align}
Eliminating $k_1$ and $k_2$ from Eqs.~\eqref{eq:proc_sys1}--\eqref{eq:proc_sys2} and solving for outputs gives
\begin{align}
    q_1 &= \frac{\displaystyle\left[a - \frac{c_1}{1-t}\right]\!\left[1 - \frac{1}{2(1-b_2)(1-s_2)\varphi_2}\right] - \frac{m(a-c_2)}{2(1-b_2)}}{2(1-b_1)\,\Delta^k}, \label{eq:q1_proc}\\[8pt]
    q_2 &= \frac{\displaystyle\left[a - c_2\right]\!\left[1 - \frac{1}{2(1-b_1)(1-t)(1-s_1)\varphi_1}\right] - \frac{m}{2(1-b_1)}\!\left[a - \frac{c_1}{1-t}\right]}{2(1-b_2)\,\Delta^k}, \label{eq:q2_proc}
\end{align}
where
\begin{equation}\label{eq:Delta_k}
\begin{split}
        \Delta^k & \;\equiv\; \left[1 - \frac{1}{2(1-b_1)(1-t)(1-s_1)\varphi_1}\right]\!\left[1 - \frac{1}{2(1-b_2)(1-s_2)\varphi_2}\right] \\
        & - \frac{m^2}{4(1-b_1)(1-b_2)}.
    \end{split}
\end{equation}
An interior equilibrium with strictly positive quantities requires $\Delta^k > 0$ and the two firms' second-order conditions, $2(1-b_1)(1-t)(1-s_1)\varphi_1 > 1$ and $2(1-b_2)(1-s_2)\varphi_2 > 1$. We maintain these conditions throughout the analysis (see Appendix~\ref{app:derivations_process}).

The indirect effect is visible in the equilibrium quantities~\eqref{eq:q1_general}--\eqref{eq:Delta_general}, where the externality in consumption of each good enters the expressions for both quantities for $m \neq 0$. The indirect externality works in the same direction as the direct externality when the goods are complements and in the opposite direction when the goods are substitutes. The reason is that a stronger network externality in one good raises the willingness to pay for that good and therefore its consumption. Higher consumption of one good raises the demand for a complement and reduces the demand for a substitute.

Table~\ref{tab:statics} summarises the signs of the Stage~2 comparative statics for substitutes ($m>0$). The import tax reduces the foreign firm's output and R\&D. When the goods are substitutes, a reduction in foreign output reduces competitive pressure on the home firm, thereby increasing both home output and R\&D investment. The business-stealing effect operates through the subsidies. A subsidy increases the firm's output and R\&D and reduces the rival's. A stronger network externality raises the willingness to pay generated by each additional unit of consumption, making the good's price less responsive to increases in output. As a result, output and R\&D increase for its producer but decrease for its rival.

\begin{table}[ht]
\centering
\caption{Signs of the Stage-2 comparative statics for substitutes ($m > 0$)}
\label{tab:statics}
\renewcommand{\arraystretch}{1.3}
\begin{tabular}{@{}lccccc@{}}
\toprule
 & \multicolumn{5}{c}{Effect of an increase in} \\
\cmidrule(l){2-6}
 & $t$ & own subsidy & rival's subsidy & own $b$ & rival's $b$ \\
\midrule
Foreign output and R\&D ($q_1$, $x_1$) & $-$ & $+$ & $-$ & $+$ & $-$ \\
Home output and R\&D ($q_2$, $x_2$) & $+$ & $+$ & $-$ & $+$ & $-$ \\
\bottomrule
\end{tabular}
\parbox{0.92\textwidth}{\footnotesize\medskip Notes: $x_i$ denotes $k_i$ under process R\&D and $r_i$ under product R\&D. The sign patterns are identical under the two specifications (Propositions~\ref{propA:policy_process}, \ref{propA:network_process}, \ref{propB:policy_product} and~\ref{propB:network_product}). The effect of $t$ on the foreign firm is direct and holds for all $m$. The remaining effects in the $t$ column and the columns for the rival's subsidy and the rival's $b$ are cross-effects operating through $m$. They vanish at $m = 0$ and reverse sign for complements ($m < 0$).}
\end{table}

Two features are important for the analysis that follows. First, cross-effects operate through $m$, implying that they vanish when the goods are independent and reverse sign when the goods are complements. Second, network externalities amplify the response of output. Holding the rival's output fixed, a shock that would change a firm's output by one unit in the absence of the externality changes it by $1/(1-b_i)$ units in the model, before accounting for the firm's endogenous R\&D response. This amplification arises because a reduction in output lowers willingness to pay, which in turn induces a further reduction in the firm's output.\footnote{Because the model is static, this multiplier operates within a single equilibrium and does not represent a dynamic process.}

We now turn to the Stage~1 policy choices. The foreign government's problem admits a closed-form best response solution.

\begin{proposition}[Foreign subsidy best response under process R\&D]\label{prop:foreign_subsidy_process}
The foreign government's best response in the process R\&D subsidy competition is:
\begin{equation}\label{eq:s1_star}
    s_1^* = \frac{\dfrac{m^2}{4(1-b_1)(1-b_2)}}{1 - \dfrac{1}{2(1-b_2)(1-s_2)\varphi_2}}.
\end{equation}
In particular,
\begin{enumerate}
    \item[\textup{(i)}] $s_1^* > 0$ for $m \neq 0$: the foreign government subsidises its firm unless the goods are independent.
    \item[\textup{(ii)}] $s_1^*$ is independent of the import tax $t$.
    \item[\textup{(iii)}] $s_1^*$ is increasing in the home subsidy $s_2$ (i.e.\ $ds_1^*/ds_2 > 0$ for $m \neq 0$): subsidies are strategic complements.
\end{enumerate}
\end{proposition}

\begin{proof}
See Appendix~\ref{app:derivations_process}.
\end{proof}

The foreign subsidy $s_1^*$ is independent of $t$ because the import tax affects the marginal benefit and the marginal cost of the subsidy in the same proportion. Consequently, the subsidy that equates these two remains unchanged. The strategic complementarity ($ds_1^*/ds_2 > 0$) implies that a home subsidy induces a higher foreign subsidy. This is consistent with the sequence of subsidy commitments observed in the markets for semiconductors and EVs.

\begin{proposition}[Home subsidy best response under process R\&D and independent goods]\label{prop:home_policy_process}
When goods are independent ($m = 0$), the home government's best response in the process R\&D subsidy is:
\begin{equation}\label{eq:s2_star}
    s_2^* = \frac{1}{1 + 2(1-b_2)}.
\end{equation}
In particular, the optimal home subsidy is strictly positive, is independent of the foreign subsidy, and is increasing in the network externality parameter $b_2$.
\end{proposition}

\begin{proof}
See Appendix~\ref{app:derivations_process}.
\end{proof}

For $m > 0$, the home government's first-order conditions for $s_2$ and $t$ at an interior solution form a non-linear system that cannot be solved in closed form (Appendix~\ref{app:derivations_process}, Section~\ref{app:proc_home}). The numerical exercises in Section~\ref{sec:simulations} confirm that, throughout the admissible range of parameters in an interior equilibrium, the home government optimally sets a positive subsidy $s_2 > 0$; the incentive to subsidise domestic R\&D is greater when the network externality is stronger. The optimal tax $t$ in an interior equilibrium can be either positive or negative; see Figure~\ref{fig:winwin_map} for illustration. 

\subsection{Product R\&D}\label{sec:productRD}

With product R\&D, the structure of the interaction changes in two respects. First, the product R\&D directly increases willingness to pay, further amplified by the network externality, while process R\&D affects the demand only indirectly through its effect on the firm's output. Second, with product R\&D the import tax enters the foreign firm's product R\&D decision directly, whereas under process R\&D the import tax works indirectly, through the equilibrium output.

Setting $k_1 = k_2 = 0$ gives the Stage~2 equilibrium conditions
\begin{align}
    q_1 &= \frac{1}{2(1-b_1)}\left[a + r_1 - mq_2 - \frac{c_1}{1-t}\right], &
    r_1 &= \frac{(1-t)\,q_1}{(1-\sigma_1)\theta_1}, \label{eq:prod_sys1}\\[4pt]
    q_2 &= \frac{1}{2(1-b_2)}\left[a + r_2 - mq_1 - c_2\right], &
    r_2 &= \frac{q_2}{(1-\sigma_2)\theta_2}. \label{eq:prod_sys2}
\end{align}
In Eq.~\eqref{eq:prod_sys1} the tax appears in $r_1$ as the factor $(1-t)$, so it reduces the foreign firm's return to product R\&D directly. The tax reduces both the firm's output and quality per unit sold, whereas under process R\&D it affects the firm's incentives only through output. Therefore, a given tax is more effective under product R\&D, which accounts for the differences in welfare outcomes discussed in Section~\ref{sec:simulations}.

Eliminating $r_1$ and $r_2$ from Eqs.~\eqref{eq:prod_sys1}--\eqref{eq:prod_sys2} and solving for outputs gives the equilibrium quantities:
\begin{align}
    q_1 &= \frac{\displaystyle\left[a - \frac{c_1}{1-t}\right]\!\left[1 - \frac{1}{2(1-b_2)(1-\sigma_2)\theta_2}\right] - \frac{m(a-c_2)}{2(1-b_2)}}{2(1-b_1)\,\Delta^r}, \label{eq:q1_prod}\\[8pt]
    q_2 &= \frac{\displaystyle\left[a - c_2\right]\!\left[1 - \frac{1-t}{2(1-b_1)(1-\sigma_1)\theta_1}\right] - \frac{m}{2(1-b_1)}\!\left[a - \frac{c_1}{1-t}\right]}{2(1-b_2)\,\Delta^r}, \label{eq:q2_prod}
\end{align}
where
\begin{equation}\label{eq:Delta_r}
    \Delta^r \;\equiv\; \left[1 - \frac{1-t}{2(1-b_1)(1-\sigma_1)\theta_1}\right]\!\left[1 - \frac{1}{2(1-b_2)(1-\sigma_2)\theta_2}\right] - \frac{m^2}{4(1-b_1)(1-b_2)}.
\end{equation}
An interior equilibrium requires each of the two bracketed factors in Eq.~\eqref{eq:Delta_r} to be positive (these are each firm's second-order conditions), together with $\Delta^r > 0$, analogous to the conditions in Section~\ref{sec:processRD} (see Appendix~\ref{app:derivations_product}). The comparative statics mirror the process-R\&D case, with the same sign pattern as in Table~\ref{tab:statics}. The formal statements and proofs are Propositions~\ref{propB:policy_product} and~\ref{propB:network_product} in Appendix~\ref{app:derivations_product}.

The Stage~1 best responses of the governments take the same form as under process R\&D.

\begin{proposition}[Foreign subsidy best response under product R\&D]\label{prop:foreign_subsidy_product}
The foreign government's best response in the product R\&D subsidy is:
\begin{equation}\label{eq:sigma1_star}
    \sigma_1^* = \frac{\dfrac{m^2}{4(1-b_1)(1-b_2)}}{1 - \dfrac{1}{2(1-b_2)(1-\sigma_2)\theta_2}}.
\end{equation}
Specifically,
\begin{enumerate}
    \item[\textup{(i)}] $\sigma_1^* > 0$ when $m \neq 0$.
    \item[\textup{(ii)}] $\sigma_1^*$ is independent of $t$.
    \item[\textup{(iii)}] Subsidies are strategic complements.
\end{enumerate}
\end{proposition}

\begin{proof}
See Appendix~\ref{app:derivations_product}.
\end{proof}

\begin{proposition}[Home policy under product R\&D and independent goods]\label{prop:home_policy_product}
When $m = 0$, the home government's optimal product R\&D subsidy is:
\begin{equation}\label{eq:sigma2_star}
    \sigma_2^* = \frac{1}{1+2(1-b_2)},
\end{equation}
which is increasing in $b_2$.
\end{proposition}

\begin{proof}
See Appendix~\ref{app:derivations_product}.
\end{proof}

For $m > 0$, the home government's optimal policy has no closed-form solution. Several representative numerical examples are shown in Section~\ref{sec:simulations}. At low values of $b$, the home government combines a positive subsidy with a positive tax, as in the process R\&D case. At higher $b$, however, the optimal tax can turn negative under product R\&D when the goods are weak substitutes or complements. This result is discussed in Section~\ref{sec:sim_map}.

\section{Welfare Effects of Industrial Policy}\label{sec:simulations}

In the absence of closed-form solutions, to compare the welfare levels between policy competition SPNE and the laissez-faire outcomes we used numerical exercises varying the strength of network externalities $b$ and the degree of product differentiation $m$ as the main parameters of interest. The patterns obtained over the admissible range of parameters are illustrated in Figures~\ref{fig:welfare_diff_process}--\ref{fig:winwin_map}. 

One can see that when the externalities are weak and the goods are moderate to strong substitutes the foreign country is worse off under policy competition than under laissez-faire throughout the admissible range, and the home country is better off. Policy competition also reduces global welfare. As the strength of externality increases, the home country's gain can outweigh the foreign country's loss, and the aggregate welfare in SPNE then exceeds the laissez-faire level (Section~\ref{sec:sim_process}). Moreover, both countries can be better off under policy competition than under laissez-faire when the goods are either weak substitutes or complements. The optimal import tax can turn negative and become an import subsidy. The gain in consumer welfare, amplified by corrected externality, outweighs the home firm's profit loss from business-stealing (Section~\ref{sec:sim_product}). The higher the degree of differentiation between the home and foreign goods, the lower is the minimal strength of externality required for the win-win outcome. The set of win-win configurations of $(m, b)$ is larger under product R\&D than under process R\&D (Section~\ref{sec:sim_map}).\footnote{The optimal policy instruments as functions of $b$ are reported in Supplementary Material, Section~\ref{sm:policy_instruments}, and numerical examples for complements in Section~\ref{sm:complements}.}

\subsection{Parameterisation}\label{sec:calibration}

For the numerical illustrations we focus on interior equilibria (Section~\ref{sec:stage2}), which restricts the set of parameters we explore. Throughout the exercises, we set $a = 1$, $\varphi_1 = \varphi_2 = \theta_1 = \theta_2 = 2.5$, and $c_1 = c_2 = 0.8$. We assume the symmetric strength of network externalities for the two goods, $b_1 = b_2 = b$, where $b$ is varied from zero to the upper bound $\bar{b}$  of the admissible region (see Table \ref{tab:calibration}).\footnote{The aim of this numerical exercise is to illustrate the possible patterns, rather than to produce calibrated empirical predictions. } 

In Figures \ref{fig:welfare_diff_process} and \ref{fig:welfare_diff_product} each pattern is illustrated for two values of the substitutability parameter, $m = 0.01$ (nearly independent goods) and $m = 0.25$ (moderately substitutable goods), to show how the welfare effects of network externalities depend on the degree of product differentiation. The admissible range of $b$ depends on $m$. At $m = 0.25$ an interior equilibrium exists only for $b$ below approximately $0.40$ under both specifications, because the optimal import tax rises steeply as $b$ approaches that value and drives the foreign firm's output to zero. At $m = 0.01$ the admissible range for $b$ ends at approximately $0.50$ under process R\&D, where the home firm's post-innovation marginal cost $c_2 - k_2$ reaches zero, and at approximately $0.56$ under product R\&D, because beyond this point $\Delta^r$ in Eq.~\eqref{eq:Delta_r} falls towards zero and the quantities grow without bound. Hence in these figures we restrict $b$ between $0$ and $0.50$.

\begin{table}[ht]
\centering
\captionsetup{width=0.68\textwidth}
\caption{Parameter values used in the numerical exercises}
\label{tab:calibration}
\renewcommand{\arraystretch}{1.2}
\begin{tabular}{@{}lll@{}}
\toprule
\textbf{Parameter} & \textbf{Value} & \textbf{Meaning} \\
\midrule
$a$ & 1 & Demand intercept \\
$c_1 = c_2$ & 0.8 & Marginal cost \\
$\varphi_1 = \varphi_2$ & 2.5 & Process R\&D cost parameter \\
$\theta_1 = \theta_2$ & 2.5 & Product R\&D cost parameter \\
$b_1 = b_2 = b$ & $[0,\, \bar{b}]$ & Network externality strength \\
$m$ & $\{0.01,\, 0.25\}$; $[-0.5,\, 0.5]$ in Figure~\ref{fig:winwin_map} & Degree of substitutability \\
\bottomrule
\end{tabular}
\end{table}

\subsection{Process R\&D}\label{sec:sim_process}

Figure~\ref{fig:welfare_diff_process} shows the welfare difference between policy competition and laissez-faire equilibria, $\Delta W_i \equiv W_i^{\text{SPNE}} - W_i^{\text{LF}}$, under process R\&D. The left panel shows that for weak substitutes ($m = 0.01$; solid line) the foreign country is worse off under policy competition when the externality is weak ($\Delta W_1 < 0$); the loss is small, it shrinks towards zero as $b$ increases,  and eventually it turns positive: when externalities are sufficiently strong, the foreign country gains from policy competition. For moderate substitutes ($m = 0.25$; dashed line) the foreign country is worse off throughout the admissible range of $b$, and the loss steadily increases in $b$. The middle panel shows that the home country gains ($\Delta W_2 > 0$) throughout the admissible range of $b$, with the gain increasing in $b$ and larger for weaker substitutes. The right panel shows that the aggregate welfare difference $\Delta W$ changes sign as $b$ increases. When $b$ is low (weak externality), policy competition leads to the aggregate welfare loss, because the foreign welfare loss dominates the home welfare gain. As $b$ rises, the home gain grows faster than the foreign loss and $\Delta W$ turns positive; the threshold $b$ is higher when the goods are closer substitutes. One can also see that for weaker substitutes the aggregate welfare gain is higher than for stronger substitutes. The welfare is higher in the policy competition equilibrium because the R\&D subsidies reduce distortion by increasing consumption. The correction is larger for complements (mutually reinforcing increase in consumption) and for stronger externalities (higher amplification of the network effect).

\begin{figure}[ht]
    \centering
    \includegraphics[width=\textwidth]{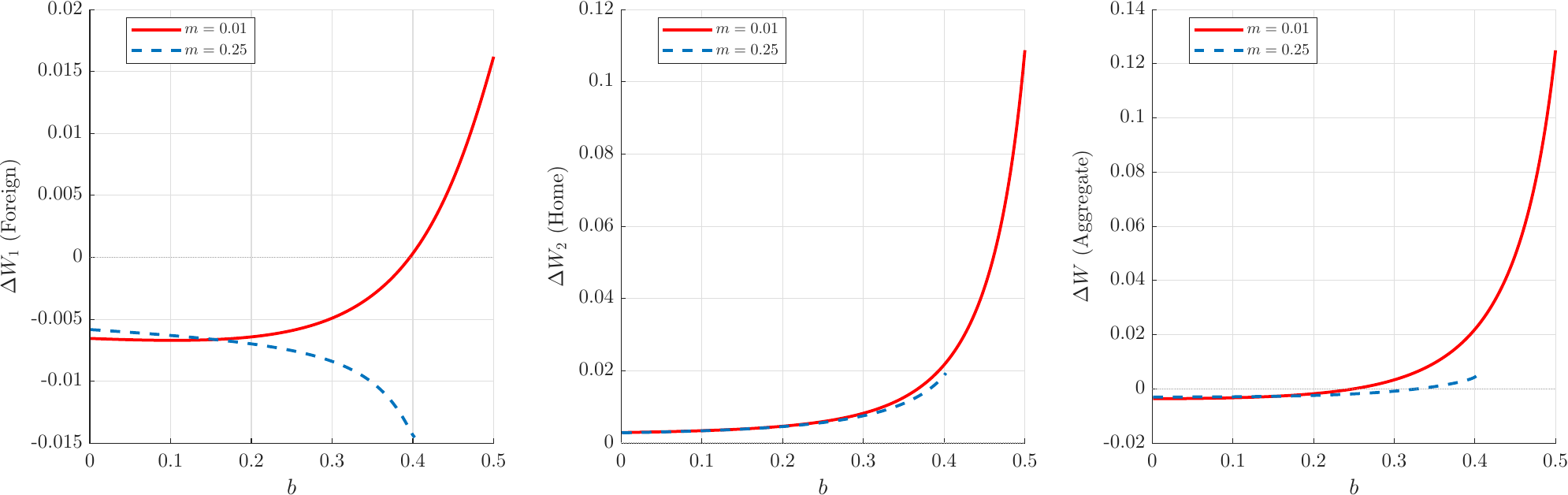}
    \caption{Process R\&D: welfare differences $\Delta W_i \equiv W_i^{\text{SPNE}} - W_i^{\text{LF}}$ against network externality strength $b$, for $m = 0.01$ (solid) and $m = 0.25$ (dashed).}
    \label{fig:welfare_diff_process}
\end{figure}

\subsection{Product R\&D}\label{sec:sim_product}

Figure~\ref{fig:welfare_diff_product} reports the same welfare comparison for the product-R\&D specification. At $m = 0.25$, the product-R\&D welfare differences are qualitatively similar to those in the process-R\&D case. The foreign country loses throughout, the home country gains, and the aggregate welfare difference turns positive only when $b$ is near the upper boundary of the admissible range. Stronger substitutability makes industrial policy redistributive under both R\&D specifications. By contrast, at $m = 0.01$, the foreign country's loss decreases as $b$ increases and eventually turns into a gain. Thus, when externality strength exceeds the threshold value of $b$ (at which $\Delta W_1$ crosses zero) the foreign country is better off under policy competition than under laissez-faire. The home country also gains over this range, so both countries benefit from policy competition relative to laissez-faire. The aggregate welfare gain $\Delta W$ is substantial at high $b$. The mechanism is the same as under the process R\&D, but the effect is more pronounced because the product R\&D increases the demand directly, whereas the process R\&D operates through the industry supply.

\begin{figure}[ht]
    \centering
    \includegraphics[width=\textwidth]{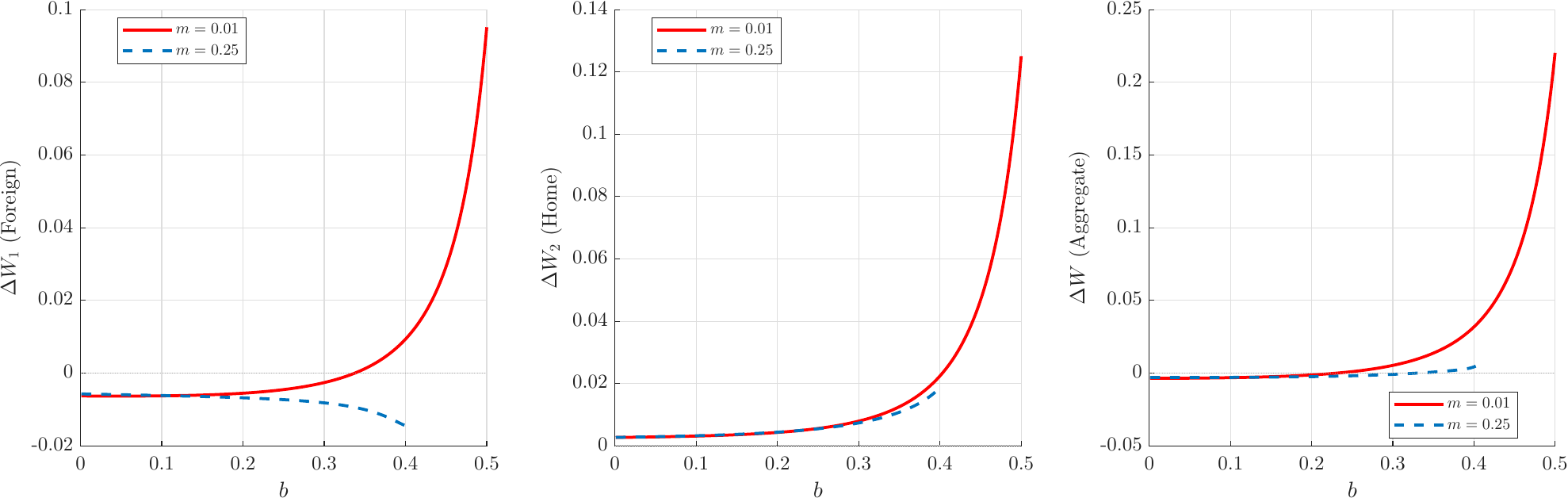}
    \caption{Product R\&D: welfare differences $\Delta W_i \equiv W_i^{\text{SPNE}} - W_i^{\text{LF}}$ against network externality strength $b$, for $m = 0.01$ (solid) and $m = 0.25$ (dashed).}
    \label{fig:welfare_diff_product}
\end{figure}

\subsection{The Win-Win Region in the $(m, b)$ Space}\label{sec:sim_map}

To illustrate that the win-win outcome exists in a non-trivial, continuous set of parameter configurations, we map the sign of the two countries' welfare gains over a grid in the $(m, b)$ plane shown in Figure~\ref{fig:winwin_map}. In this exercise we compute the interior policy competition and laissez-faire equilibria at each point of the $(m, b)$ grid for the same values of other parameters as in the previous sections.\footnote{Grid cells close to the boundary of the admissible region are computationally unstable and are, therefore, excluded.}

\begin{figure}[ht]
    \centering
    \begin{subfigure}[b]{0.49\textwidth}
        \includegraphics[width=\textwidth]{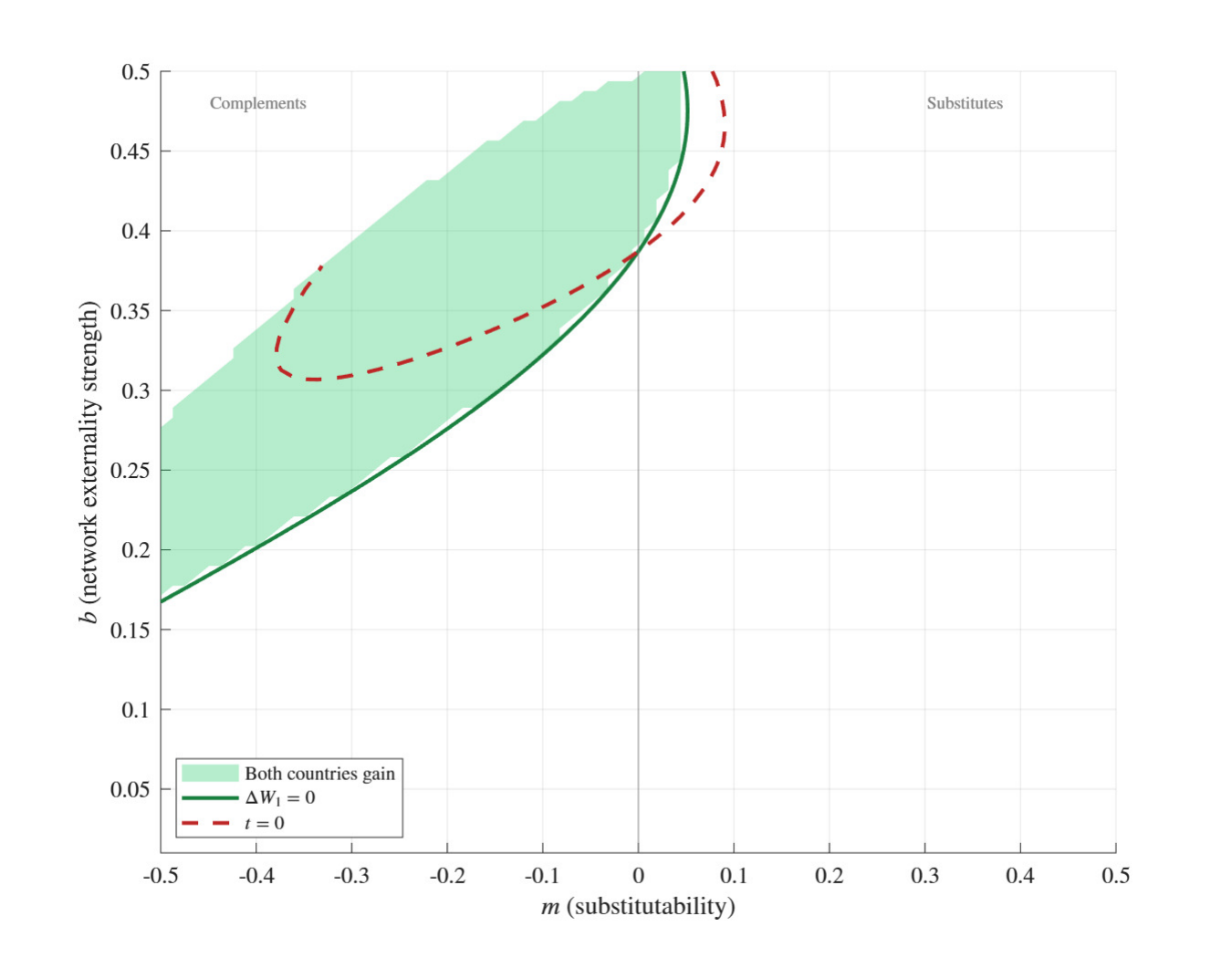}
        \caption{Process R\&D}
        \label{fig:winwin_process}
    \end{subfigure}
    \hfill
    \begin{subfigure}[b]{0.49\textwidth}
        \includegraphics[width=\textwidth]{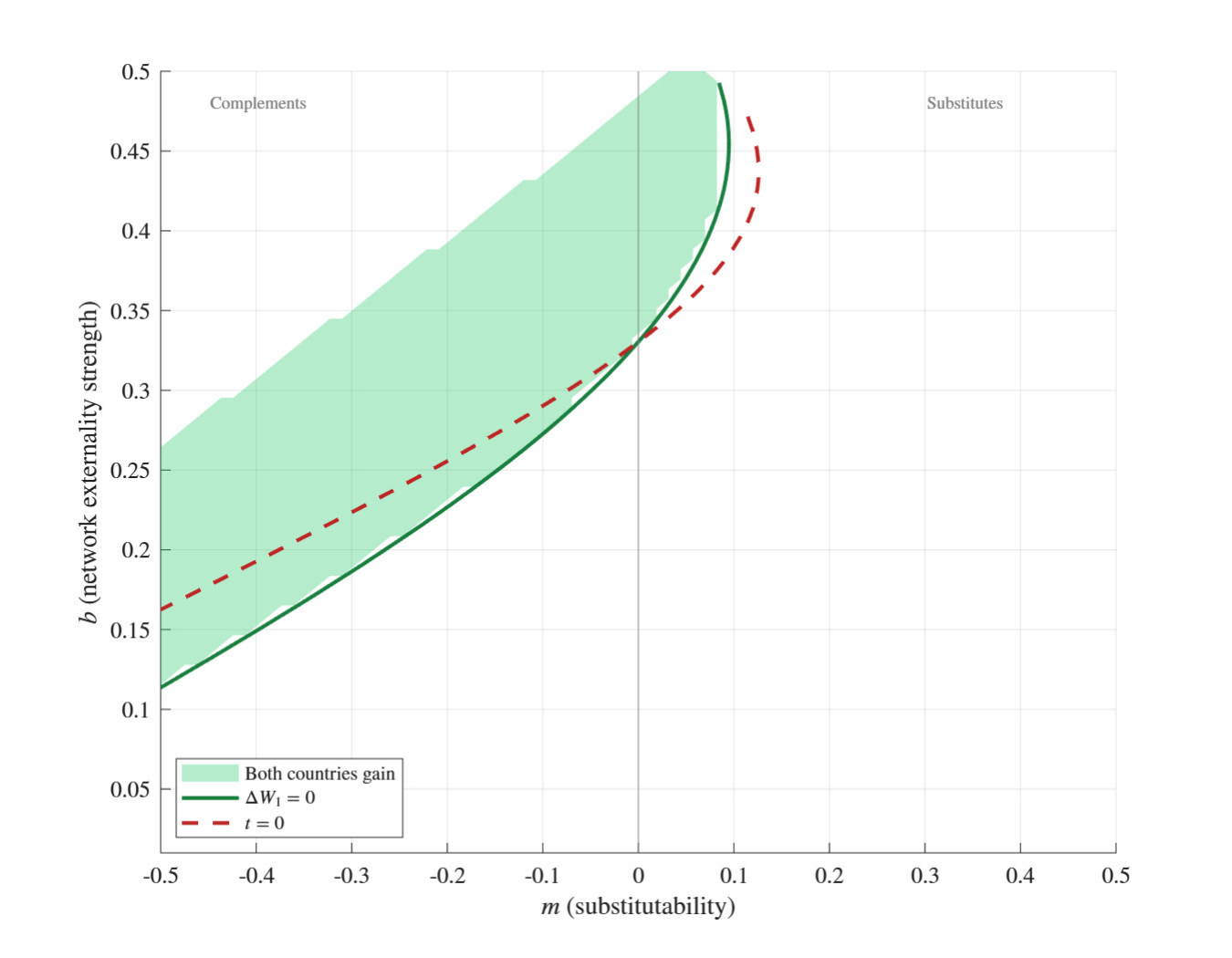}
        \caption{Product R\&D}
        \label{fig:winwin_product}
    \end{subfigure}
    \caption{The win-win region in $(m, b)$ space. Shading marks the combinations of substitutability $m$ and network externality strength $b$ at which both countries are better off under non-cooperative policy competition than under laissez-faire. The solid line is the foreign country's break-even locus ($\Delta W_1 = 0$) and the dashed line is the free-trade locus ($t = 0$), above which the optimal import tax is negative. The region is bounded from above by the parameter admissibility frontier.}
    \label{fig:winwin_map}
\end{figure}

The solid line on each panel of Figure~\ref{fig:winwin_map} is the locus of the points where $\Delta W_1 = 0$; above this line $\Delta W_1 > 0$. The upper boundary of the shaded area lies at or below the frontier of the admissible set of $(m,b)$ (the interior equilibria are below the frontier). Hence, the shaded area shows admissible configurations of $(m,b)$ where the foreign country is better off in equilibrium with policy competition relative to laissez-faire. Since the home country is always better off under policy competition, the shaded areas on each panel show the win-win set. The dashed line on each panel is the locus of the points where the optimal import tax is zero. Above this line $t<0$: the home country subsidises imports.  

One can see that the patterns are consistent with the intuition and the results illustrated in Figures \ref{fig:welfare_diff_process} and \ref{fig:welfare_diff_product} for two specific values of $m$. Under product R\&D (right panel), the win-win region covers the entire admissible range of $m<0$ and a non-trivial range of $m>0$ as long as $b$ is sufficiently high; the threshold increases in $m$. 

Under process R\&D (left panel) the picture is similar, but the range of $m>0$ where win-win equilibria exist is noticeably smaller and the threshold for $b$ is higher than under product R\&D. Intuitively, when the two goods are complements, higher consumption of one good, boosted by higher quality (product R\&D) or lower production cost leading to lower price (process R\&D), increases the consumption of the other good. Positive externalities in consumption mutually reinforce this effect for both goods. Thus, R\&D subsidies correct the under-consumption, and increase welfare in both countries when the externalities are sufficiently strong. As the product differentiation decreases, the effect becomes smaller and eventually mutual reinforcement disappears, but sufficiently strong externalities can still support the corrective effect of subsidies even for substitutes, as long as the substitutability is sufficiently weak. The effect is more pronounced for product R\&D because it increases directly the demand, or the willingness-to-pay, whereas the process R\&D operates indirectly through the industry supply.

Figure~\ref{fig:winwin_map} also demonstrates the role of the import tax. Under product R\&D (right panel) the optimal import tax is positive when externalities are relatively weak. As externalities strengthen, the tax rate tends to zero and, for complements and weak substitutes, eventually turns negative. The threshold is lower, the more differentiated the goods are. This threshold, however, does not coincide with the externality threshold for the win-win equilibrium. One can see from the graph that the win-win situation is consistent with an import tax (shaded area between the solid and the dashed lines) or with an import subsidy (shaded area above the higher of the two lines) when the goods are complements. 

This suggests that both the importer and exporter can benefit from competition in industrial policies even when the importer implements protectionist import taxes, as long as the goods are sufficiently differentiated. Not taxing or even subsidising imports can also be mutually beneficial, but this requires a stronger externality, so that the welfare increase from correcting the externality outweighs the fiscal cost for the home country. When the goods are substitutes, for both countries to benefit from industrial policy competition the imports should be subsidised.

Under process R\&D (left panel), since the corrective effect of R\&D subsidies is weaker, the home country taxes imports even for relatively strong externalities, and the set of $(m,b)$ where an import subsidy is optimal is relatively small and close to the frontier of the admissible range. Nevertheless, these results illustrate the possibility of the win-win outcome consistent with protectionist taxation of imports when the network goods are sufficiently differentiated.

\section{Conclusion}\label{sec:conclusion}

The recent proliferation of industrial and protectionist policies among major economies is frequently criticised as being damaging and self-defeating. Our analysis demonstrates that this characterisation is overly broad, particularly in modern high-technology sectors that feature strong network externalities. In these environments, industrial policies combining targeted R\&D subsidies for domestic producers with trade instruments can generate mutually beneficial outcomes, provided that consumption externalities are sufficiently strong and goods are sufficiently differentiated. When the goods are substitutes, achieving a win-win outcome requires import subsidies rather than protective tariffs. 

Our results also suggest that the potential for mutually beneficial industrial policy competition depends on the nature of the innovation being subsidised. The welfare gains are significantly stronger when governments support product innovation than process innovation (examples of the former are not limited to designing new product features and can include developing next-generation technology standards or system platforms; examples of the latter include factory automation and manufacturing yield improvements). This divergence occurs because product innovation acts directly on the demand side, raising baseline consumer valuation and immediately reinforcing adoption networks, whereas process innovation works indirectly through cost reduction and supply expansion. Ultimately, the welfare effect of industrial policy competition is governed by the structural interaction among demand-side network externalities, product differentiation, and the type of innovation.

We derived these results within a stylised two-country, two-good framework focussing on domestic consumption and single-product network externalities. While closed-form solutions for welfare comparisons between policy competition and laissez-faire equilibria could not be obtained, extensive numerical evaluations confirm that these qualitative policy patterns are robust across a wide range of parameter values (Supplementary Material, Section~\ref{sm:robustness}). Extending the setup to incorporate foreign consumption or cross-product consumption externalities would allow introducing cross-border network spillovers without altering the fundamental mechanisms driving our results.

From a policy perspective, our findings suggest that both indiscriminate protectionism and uniform critiques of industrial intervention can be misguided. Industrial policies are neither inherently value-creating nor universally destructive, nor are they a zero-sum game. Their efficacy should be assessed in the context of the markets, taking into account sector-specific adoption spillovers, similarity of products or complementarities between products and infrastructures, and the targets of R\&D subsidies.

\section*{Declaration of generative AI and AI-assisted technologies in the manuscript preparation process}

During the preparation of this work the authors used Claude (Anthropic) for editing and proofreading the text and for assistance with the code for the numerical simulations. The authors reviewed and edited the output as necessary and take full responsibility for the content of the finished article.

\section*{Acknowledgements}

This work was supported by the Economic and Social Research Council [grant number ES/P000762/1].
\clearpage

%% file: appendix.tex
\section{Derivations for Process R\&D Equilibrium}\label{app:derivations_process}

This appendix provides the detailed derivations underlying the results in Section~\ref{sec:processRD}, together with the formal statements of the comparative-statics propositions summarised in Table~\ref{tab:statics}. Set $r_1 = r_2 = 0$ throughout.

\subsection{Second-order conditions and maintained assumptions}\label{app:proc_soc}

The second-order conditions for the foreign firm require:
\[
\frac{\partial^2 \pi_1}{\partial q_1^2} < 0, \qquad
\frac{\partial^2 \pi_1}{\partial k_1^2} < 0, \qquad
\Delta_1^{k*} \equiv
\begin{vmatrix}
\dfrac{\partial^2 \pi_1}{\partial q_1^2} & \dfrac{\partial^2 \pi_1}{\partial q_1 \partial k_1} \\[8pt]
\dfrac{\partial^2 \pi_1}{\partial q_1 \partial k_1} & \dfrac{\partial^2 \pi_1}{\partial k_1^2}
\end{vmatrix} > 0.
\]
Evaluating:
\[
\Delta_1^{k*} = 2(1-b_1)(1-t)(1-s_1)\varphi_1 \left[1 - \frac{1}{2(1-b_1)(1-t)(1-s_1)\varphi_1}\right].
\]
For this to be positive, we need:
\begin{equation}\label{eq:soc_k1}
2(1-b_1)(1-t)(1-s_1)\varphi_1 > 1.
\end{equation}
Similarly, for the home firm:
\begin{equation}\label{eq:soc_k2}
2(1-b_2)(1-s_2)\varphi_2 > 1.
\end{equation}
Conditions \eqref{eq:soc_k1} and \eqref{eq:soc_k2} do not hold for all parameter values. They place a lower bound on the R\&D cost parameters relative to the network externality strength and the policy instruments. We maintain both conditions, together with $\Delta^k > 0$ and positive post-innovation marginal cost, $c_i - k_i > 0$ (Section~\ref{sec:firms}), throughout the analysis. At the parameterisation of Section~\ref{sec:calibration} they hold on the entire admissible range reported.

\subsection{Matrix-form equilibrium system}\label{app:proc_matrix}

The equilibrium conditions \eqref{eq:proc_sys1}--\eqref{eq:proc_sys2} can be written in matrix form, with $(q_1, k_1, q_2, k_2)$ as the unknowns:
\[
\begin{pmatrix}
1 & -\dfrac{1}{2(1-b_1)(1-t)} & \dfrac{m}{2(1-b_1)} & 0 \\[6pt]
-\dfrac{1}{(1-s_1)\varphi_1} & 1 & 0 & 0 \\[6pt]
\dfrac{m}{2(1-b_2)} & 0 & 1 & -\dfrac{1}{2(1-b_2)} \\[6pt]
0 & 0 & -\dfrac{1}{(1-s_2)\varphi_2} & 1
\end{pmatrix}
\begin{pmatrix} q_1 \\ k_1 \\ q_2 \\ k_2 \end{pmatrix}
=
\begin{pmatrix}
\dfrac{1}{2(1-b_1)}\!\left[a - \dfrac{c_1}{1-t}\right] \\[6pt] 0 \\[6pt]
\dfrac{1}{2(1-b_2)}[a - c_2] \\[6pt] 0
\end{pmatrix}.
\]

\subsection{Comparative statics with respect to policy variables}\label{app:proc_cs_policy}

\begin{appxprop}[Policy effects on output under process R\&D]\label{propA:policy_process}
Under the maintained assumptions of Section~\ref{app:proc_soc}, the following comparative statics hold.

\medskip
\noindent\textbf{Foreign firm:}
\begin{enumerate}
    \item[\textup{(i)}] $\dfrac{\partial q_1}{\partial t} < 0$, \; $\dfrac{\partial k_1}{\partial t} < 0$: \; the import tax reduces foreign output and process R\&D.
    \item[\textup{(ii)}] $\dfrac{\partial q_1}{\partial s_1} > 0$, \; $\dfrac{\partial k_1}{\partial s_1} > 0$: \; the foreign subsidy increases foreign output and process R\&D.
    \item[\textup{(iii)}] $\dfrac{\partial q_1}{\partial s_2} \leq 0$, \; $\dfrac{\partial k_1}{\partial s_2} \leq 0$ for $m \geq 0$, \; so the home subsidy reduces foreign output and process R\&D when the goods are substitutes, with equality when $m = 0$.
\end{enumerate}

\medskip
\noindent\textbf{Home firm:}
\begin{enumerate}
    \item[\textup{(iv)}] $\dfrac{\partial q_2}{\partial t} \geq 0$, \; $\dfrac{\partial k_2}{\partial t} \geq 0$ for $m \geq 0$: \; the import tax increases home output and R\&D when goods are substitutes.
    \item[\textup{(v)}] $\dfrac{\partial q_2}{\partial s_2} > 0$, \; $\dfrac{\partial k_2}{\partial s_2} > 0$: \; the home subsidy always increases home output and R\&D.
    \item[\textup{(vi)}] $\dfrac{\partial q_2}{\partial s_1} \leq 0$, \; $\dfrac{\partial k_2}{\partial s_1} \leq 0$ for $m \geq 0$: \; the foreign subsidy reduces home output and process R\&D when goods are substitutes.
\end{enumerate}
\end{appxprop}

\begin{proof}
Total differentiation of the equilibrium system with respect to $(t, s_1, s_2)$ gives:
\[
\begin{pmatrix}
1 & -\dfrac{1}{2(1-b_1)(1-t)} & \dfrac{m}{2(1-b_1)} & 0 \\[6pt]
-\dfrac{1}{(1-s_1)\varphi_1} & 1 & 0 & 0 \\[6pt]
\dfrac{m}{2(1-b_2)} & 0 & 1 & -\dfrac{1}{2(1-b_2)} \\[6pt]
0 & 0 & -\dfrac{1}{(1-s_2)\varphi_2} & 1
\end{pmatrix}
\begin{pmatrix} dq_1 \\ dk_1 \\ dq_2 \\ dk_2 \end{pmatrix}
=
\begin{pmatrix}
-\dfrac{c_1-k_1}{2(1-b_1)(1-t)^2}\,dt \\[6pt]
\dfrac{q_1}{(1-s_1)^2\varphi_1}\,ds_1 \\[6pt]
0 \\[6pt]
\dfrac{q_2}{(1-s_2)^2\varphi_2}\,ds_2
\end{pmatrix}.
\]
Applying Cramer's rule, the denominator is $\Delta^k$ as defined in Eq.~\eqref{eq:Delta_k}. The numerators yield:

\paragraph{Effects of the import tax $t$:}
\begin{align*}
\frac{\partial q_1}{\partial t} &= -\frac{1}{\Delta^k}\,\frac{c_1 - k_1}{2(1-b_1)(1-t)^2}\left[1 - \frac{1}{2(1-b_2)(1-s_2)\varphi_2}\right] < 0, \\[4pt]
\frac{\partial k_1}{\partial t} &= -\frac{1}{\Delta^k}\,\frac{c_1 - k_1}{2(1-b_1)(1-t)^2(1-s_1)\varphi_1}\left[1 - \frac{1}{2(1-b_2)(1-s_2)\varphi_2}\right] < 0, \\[4pt]
\frac{\partial q_2}{\partial t} &= \frac{1}{\Delta^k}\,\frac{m(c_1 - k_1)}{4(1-b_1)(1-b_2)(1-t)^2} \;\gtrless\; 0 \;\text{ for }\; m \;\gtrless\; 0, \\[4pt]
\frac{\partial k_2}{\partial t} &= \frac{1}{\Delta^k}\,\frac{m(c_1 - k_1)}{4(1-b_1)(1-b_2)(1-t)^2(1-s_2)\varphi_2} \;\gtrless\; 0 \;\text{ for }\; m \;\gtrless\; 0.
\end{align*}

\paragraph{Effects of the foreign subsidy $s_1$:}
\begin{align*}
\frac{\partial q_1}{\partial s_1} &= \frac{1}{\Delta^k}\,\frac{q_1}{2(1-b_1)(1-t)(1-s_1)^2\varphi_1}\left[1 - \frac{1}{2(1-b_2)(1-s_2)\varphi_2}\right] > 0, \\[4pt]
\frac{\partial k_1}{\partial s_1} &= \frac{1}{\Delta^k}\,\frac{q_1}{(1-s_1)^2\varphi_1}\left[1 - \frac{1}{2(1-b_2)(1-s_2)\varphi_2} - \frac{m^2}{4(1-b_1)(1-b_2)}\right] > 0, \\[4pt]
\frac{\partial q_2}{\partial s_1} &= -\frac{1}{\Delta^k}\,\frac{mq_1}{4(1-b_1)(1-b_2)(1-t)(1-s_1)^2\varphi_1} \;\lessgtr\; 0 \;\text{ for }\; m \;\gtrless\; 0, \\[4pt]
\frac{\partial k_2}{\partial s_1} &= -\frac{1}{\Delta^k}\,\frac{mq_1}{4(1-b_1)(1-b_2)(1-t)(1-s_1)^2\varphi_1(1-s_2)\varphi_2} \;\lessgtr\; 0 \;\text{ for }\; m \;\gtrless\; 0.
\end{align*}

\paragraph{Effects of the home subsidy $s_2$:}
\begin{align*}
\frac{\partial q_1}{\partial s_2} &= -\frac{1}{\Delta^k}\,\frac{mq_2}{4(1-b_1)(1-b_2)(1-s_2)^2\varphi_2} \;\lessgtr\; 0 \;\text{ for }\; m \;\gtrless\; 0, \\[4pt]
\frac{\partial q_2}{\partial s_2} &= \frac{1}{\Delta^k}\,\frac{q_2}{2(1-b_2)(1-s_2)^2\varphi_2}\left[1 - \frac{1}{2(1-b_1)(1-t)(1-s_1)\varphi_1}\right] > 0, \\[4pt]
\frac{\partial k_2}{\partial s_2} &= \frac{1}{\Delta^k}\,\frac{q_2}{(1-s_2)^2\varphi_2}\left[1 - \frac{1}{2(1-b_1)(1-t)(1-s_1)\varphi_1} - \frac{m^2}{4(1-b_1)(1-b_2)}\right] > 0.
\end{align*}
Under the maintained assumptions the bracketed second-order-condition terms and $\Delta^k$ are strictly positive and $c_1 - k_1 > 0$, and the signs follow by inspection. The bracketed terms in $\partial k_1/\partial s_1$ and $\partial k_2/\partial s_2$ have the form of one second-order-condition factor minus $\frac{m^2}{4(1-b_1)(1-b_2)}$. Since the other factor is smaller than one, each is bounded below by $\Delta^k$ and is therefore positive. The remaining effect, $\partial k_1/\partial s_2$, follows from the chain rule applied to $k_1 = q_1/[(1-s_1)\varphi_1]$ and has the sign of $\partial q_1/\partial s_2$.
\end{proof}

\subsection{Comparative statics with respect to network externalities}\label{app:proc_cs_network}

\begin{appxprop}[Network externality effects on output under process R\&D]\label{propA:network_process}
Under the maintained assumptions of Section~\ref{app:proc_soc}:
\begin{enumerate}
    \item[\textup{(i)}] $\dfrac{\partial q_i}{\partial b_i} > 0$, \; $\dfrac{\partial k_i}{\partial b_i} > 0$: \; a stronger network externality on own good consumption increases own output and process R\&D.
    \item[\textup{(ii)}] $\dfrac{\partial q_i}{\partial b_j} \leq 0$, \; $\dfrac{\partial k_i}{\partial b_j} \leq 0$ for $m \geq 0$ \textup{($i \neq j$)}: \; a stronger network externality in the rival's good consumption reduces own output and R\&D when goods are substitutes.
\end{enumerate}
\end{appxprop}

\begin{proof}
Total differentiation with respect to $(b_1, b_2)$:
\begin{equation*}
    \begin{split}
& \begin{pmatrix}
1 & -\dfrac{1}{2(1-b_1)(1-t)} & \dfrac{m}{2(1-b_1)} & 0 \\[6pt]
-\dfrac{1}{(1-s_1)\varphi_1} & 1 & 0 & 0 \\[6pt]
\dfrac{m}{2(1-b_2)} & 0 & 1 & -\dfrac{1}{2(1-b_2)} \\[6pt]
0 & 0 & -\dfrac{1}{(1-s_2)\varphi_2} & 1
\end{pmatrix}
\begin{pmatrix} dq_1 \\ dk_1 \\ dq_2 \\ dk_2 \end{pmatrix} \\
& =
\begin{pmatrix}
\dfrac{a - \frac{c_1-k_1}{1-t} - mq_2}{2(1-b_1)^2}\,db_1 \\[6pt]
0 \\[6pt]
\dfrac{a - (c_2-k_2) - mq_1}{2(1-b_2)^2}\,db_2 \\[6pt]
0
\end{pmatrix}.
\end{split}
\end{equation*}

The resulting derivatives are:
\begin{align*}
\frac{\partial q_1}{\partial b_1} &= \frac{1}{\Delta^k}\,\frac{a - \frac{c_1-k_1}{1-t} - mq_2}{2(1-b_1)^2}\left[1 - \frac{1}{2(1-b_2)(1-s_2)\varphi_2}\right] > 0, \\[4pt]
\frac{\partial q_1}{\partial b_2} &= -\frac{m}{\Delta^k}\,\frac{a - (c_2-k_2) - mq_1}{4(1-b_1)(1-b_2)^2} \;\lessgtr\; 0 \;\text{ for }\; m \;\gtrless\; 0, \\[4pt]
\frac{\partial k_1}{\partial b_1} &= \frac{1}{\Delta^k}\,\frac{a - \frac{c_1-k_1}{1-t} - mq_2}{2(1-b_1)^2(1-s_1)\varphi_1}\left[1 - \frac{1}{2(1-b_2)(1-s_2)\varphi_2}\right] > 0, \\[4pt]
\frac{\partial k_1}{\partial b_2} &= -\frac{m}{\Delta^k}\,\frac{a - (c_2-k_2) - mq_1}{4(1-b_1)(1-b_2)^2(1-s_1)\varphi_1} \;\lessgtr\; 0 \;\text{ for }\; m \;\gtrless\; 0.
\end{align*}
The results for the home firm are analogous:
\begin{align*}
\frac{\partial q_2}{\partial b_2} &= \frac{1}{\Delta^k}\,\frac{a - (c_2-k_2) - mq_1}{2(1-b_2)^2}\left[1 - \frac{1}{2(1-b_1)(1-t)(1-s_1)\varphi_1}\right] > 0, \\[4pt]
\frac{\partial q_2}{\partial b_1} &= -\frac{m}{\Delta^k}\,\frac{a - \frac{c_1-k_1}{1-t} - mq_2}{4(1-b_1)^2(1-b_2)} \;\lessgtr\; 0 \;\text{ for }\; m \;\gtrless\; 0, \\[4pt]
\frac{\partial k_2}{\partial b_2} &= \frac{1}{\Delta^k}\,\frac{a - (c_2-k_2) - mq_1}{2(1-b_2)^2(1-s_2)\varphi_2}\left[1 - \frac{1}{2(1-b_1)(1-t)(1-s_1)\varphi_1}\right] > 0, \\[4pt]
\frac{\partial k_2}{\partial b_1} &= -\frac{m}{\Delta^k}\,\frac{a - \frac{c_1-k_1}{1-t} - mq_2}{4(1-b_1)^2(1-b_2)(1-s_2)\varphi_2} \;\lessgtr\; 0 \;\text{ for }\; m \;\gtrless\; 0.
\end{align*}
The numerators $a - \frac{c_1-k_1}{1-t} - mq_2 = 2(1-b_1)q_1 > 0$ and $a - (c_2-k_2) - mq_1 = 2(1-b_2)q_2 > 0$ by the firms' first-order conditions, and the signs follow.
\end{proof}

\subsection{Optimal foreign subsidy}\label{app:proc_foreign}

The foreign government maximises $W_1 = \pi_1 - s_1 \frac{\varphi_1 k_1^2}{2}$. Setting $dW_1/ds_1 = 0$:
\[
-m(1-t)\,q_1\,\frac{dq_2}{ds_1} - s_1\,\varphi_1\,k_1\,\frac{dk_1}{ds_1} = 0.
\]
Substituting the expressions for $\partial q_2/\partial s_1$ and $\partial k_1/\partial s_1$ from Section~\ref{app:proc_cs_policy} and simplifying yields:
\[
s_1^* = \frac{\dfrac{m^2}{4(1-b_1)(1-b_2)}}{1 - \dfrac{1}{2(1-b_2)(1-s_2)\varphi_2}}.
\]
The slope of the reaction function is:
\[
\frac{ds_1^*}{ds_2} = \frac{1}{\left[1 - \frac{1}{2(1-b_2)(1-s_2)\varphi_2}\right]^2}\,\frac{m^2}{8(1-b_1)(1-b_2)^2(1-s_2)^2\varphi_2} > 0 \quad\text{for } m \neq 0.
\]
This establishes Proposition~\ref{prop:foreign_subsidy_process}.

\subsection{Optimal home subsidy and tax}\label{app:proc_home}

Substituting Eq.~\eqref{eq:govbudget} into Eq.~\eqref{eq:W2} and omitting the constant $L$, the home government maximises $W_2 = CS + \pi_2 - s_2\frac{\varphi_2 k_2^2}{2} + tp_1q_1$. The first-order condition for the subsidy is:
\[
\frac{dW_2}{ds_2} = \left[q_1 + t\,\frac{c_1 - k_1}{1-t}\right]\frac{dq_1}{ds_2} - \frac{s_2\,q_2}{1-s_2}\,\frac{dk_2}{ds_2} + [m(1-t)q_1 + q_2]\,\frac{dq_2}{ds_2} = 0.
\]
Setting $m = 0$ and substituting the comparative statics yields $s_2^* = \frac{1}{1+2(1-b_2)}$, which establishes Proposition~\ref{prop:home_policy_process}.

The first-order condition for the optimal tax is:
\[
\frac{dW_2}{dt} = \left[q_1 + t\,\frac{c_1 - k_1}{1-t}\right]\frac{dq_1}{dt} + [q_2 + m(1-t)q_1]\,\frac{dq_2}{dt} - \frac{s_2\,q_2}{1-s_2}\,\frac{dk_2}{dt} + p_1 q_1 = 0.
\]
For $m \neq 0$, this equation together with $dW_2/ds_2 = 0$ forms a system that cannot be solved in closed form. The numerical solution is presented in Section~\ref{sec:simulations}.

\section{Derivations for Product R\&D Equilibrium}\label{app:derivations_product}

Set $k_1 = k_2 = 0$ throughout.

\subsection{Second-order conditions and maintained assumptions}\label{app:prod_soc}

The Hessian conditions require:
\begin{align*}
\Delta_1^{r*} &= 2(1-b_1)(1-t)(1-\sigma_1)\theta_1\left[1 - \frac{1-t}{2(1-b_1)(1-\sigma_1)\theta_1}\right] > 0, \\
\Delta_2^{r*} &= 2(1-b_2)(1-\sigma_2)\theta_2\left[1 - \frac{1}{2(1-b_2)(1-\sigma_2)\theta_2}\right] > 0.
\end{align*}
As in the process-R\&D case, these conditions do not hold for all parameter values. They require $2(1-b_1)(1-\sigma_1)\theta_1 > 1-t$ and $2(1-b_2)(1-\sigma_2)\theta_2 > 1$, which we maintain throughout together with $\Delta^r > 0$. At the parameterisation of Section~\ref{sec:calibration} they hold on the entire admissible range reported.

\subsection{Comparative statics with respect to policy variables}\label{app:prod_cs_policy}

\begin{appxprop}[Policy effects on output and product R\&D]\label{propB:policy_product}
Under the maintained assumptions of Section~\ref{app:prod_soc}:
\begin{enumerate}
    \item[\textup{(i)}] The import tax reduces foreign output and product R\&D: $\partial q_1 / \partial t < 0$, $\partial r_1 / \partial t < 0$.
    \item[\textup{(ii)}] Each firm's own-government subsidy increases its output and product R\&D.
    \item[\textup{(iii)}] When the goods are substitutes ($m > 0$), the rival's subsidy reduces own output and R\&D, and the import tax increases home output and R\&D.
    \item[\textup{(iv)}] When the goods are independent ($m = 0$), the effects in~\textup{(iii)} vanish.
\end{enumerate}
\end{appxprop}

\begin{proof}
Total differentiation of the product R\&D equilibrium system gives:

\begin{equation*}
\begin{split}
& \begin{pmatrix}
1 & -\dfrac{1}{2(1-b_1)} & \dfrac{m}{2(1-b_1)} & 0 \\[6pt]
-\dfrac{1-t}{(1-\sigma_1)\theta_1} & 1 & 0 & 0 \\[6pt]
\dfrac{m}{2(1-b_2)} & 0 & 1 & -\dfrac{1}{2(1-b_2)} \\[6pt]
0 & 0 & -\dfrac{1}{(1-\sigma_2)\theta_2} & 1
\end{pmatrix}
\begin{pmatrix} dq_1 \\ dr_1 \\ dq_2 \\ dr_2 \end{pmatrix} \\
& =
\begin{pmatrix}
-\dfrac{c_1}{2(1-b_1)(1-t)^2}\,dt \\[6pt]
-\dfrac{q_1}{(1-\sigma_1)\theta_1}\,dt + \dfrac{(1-t)q_1}{(1-\sigma_1)^2\theta_1}\,d\sigma_1 \\[6pt]
0 \\[6pt]
\dfrac{q_2}{(1-\sigma_2)^2\theta_2}\,d\sigma_2
\end{pmatrix}.
\end{split}
\end{equation*}
The denominator is $\Delta^r$ as defined in Eq.~\eqref{eq:Delta_r}. Applying Cramer's rule:

\paragraph{Effects of the import tax $t$:}
\begin{align*}
\frac{dq_1}{dt} &= -\frac{1}{\Delta^r}\,\frac{1}{2(1-b_1)}\left[\frac{q_1}{(1-\sigma_1)\theta_1} + \frac{c_1}{(1-t)^2}\right]\!\left[1 - \frac{1}{2(1-b_2)(1-\sigma_2)\theta_2}\right] < 0, \\[4pt]
\frac{dq_2}{dt} &= \frac{m}{\Delta^r}\,\frac{1}{2(1-b_2)}\left[\frac{q_1}{2(1-b_1)(1-\sigma_1)\theta_1} + \frac{c_1}{2(1-b_1)(1-t)^2}\right] \;\gtrless\; 0 \;\text{ for }\; m \;\gtrless\; 0.
\end{align*}

\paragraph{Effects of the foreign subsidy $\sigma_1$:}
\begin{align*}
\frac{dq_1}{d\sigma_1} &= \frac{1}{\Delta^r}\,\frac{(1-t)q_1}{2(1-b_1)(1-\sigma_1)^2\theta_1}\left[1 - \frac{1}{2(1-b_2)(1-\sigma_2)\theta_2}\right] > 0, \\[4pt]
\frac{dq_2}{d\sigma_1} &= -\frac{m}{\Delta^r}\,\frac{(1-t)q_1}{4(1-b_1)(1-b_2)(1-\sigma_1)^2\theta_1} \;\lessgtr\; 0 \;\text{ for }\; m \;\gtrless\; 0.
\end{align*}

\paragraph{Effects of the home subsidy $\sigma_2$:}
\begin{align*}
\frac{dq_1}{d\sigma_2} &= -\frac{1}{\Delta^r}\,\frac{q_2}{(1-\sigma_2)^2\theta_2}\,\frac{m}{4(1-b_1)(1-b_2)} \;\lessgtr\; 0 \;\text{ for }\; m \;\gtrless\; 0, \\[4pt]
\frac{dq_2}{d\sigma_2} &= \frac{1}{\Delta^r}\,\frac{q_2}{2(1-b_2)(1-\sigma_2)^2\theta_2}\left[1 - \frac{1-t}{2(1-b_1)(1-\sigma_1)\theta_1}\right] > 0.
\end{align*}

The R\&D investment effects follow from $r_1 = (1-t)q_1/[(1-\sigma_1)\theta_1]$ and $r_2 = q_2/[(1-\sigma_2)\theta_2]$. The expressions for $dr_1/dt$, $dr_1/d\sigma_1$ and $dr_2/d\sigma_2$ are:
\begin{align*}
\frac{dr_1}{dt} &= -\frac{1}{(1-\sigma_1)\theta_1}\!\left[q_1 + \frac{1}{\Delta^r}\,\frac{1-t}{2(1-b_1)}\!\left(\frac{q_1}{(1-\sigma_1)\theta_1} + \frac{c_1}{(1-t)^2}\right)\!\left(1 - \frac{1}{2(1-b_2)(1-\sigma_2)\theta_2}\right)\right] < 0, \\[4pt]
\frac{dr_1}{d\sigma_1} &= \frac{(1-t)q_1}{(1-\sigma_1)^2\theta_1}\left[1 + \frac{1}{\Delta^r}\,\frac{1-t}{2(1-b_1)(1-\sigma_1)\theta_1}\!\left(1 - \frac{1}{2(1-b_2)(1-\sigma_2)\theta_2}\right)\right] > 0, \\[4pt]
\frac{dr_2}{d\sigma_2} &= \frac{1}{\Delta^r}\,\frac{q_2}{(1-\sigma_2)^2\theta_2}\left[1 - \frac{1-t}{2(1-b_1)(1-\sigma_1)\theta_1} - \frac{m^2}{4(1-b_1)(1-b_2)}\right] > 0.
\end{align*}
Under the maintained assumptions every bracketed term and $\Delta^r$ are strictly positive, and the signs follow by inspection. The brackets in $dr_1/dt$ and $dr_1/d\sigma_1$ are sums of positive terms, and the bracket in $dr_2/d\sigma_2$ has the form of one second-order-condition factor minus $\frac{m^2}{4(1-b_1)(1-b_2)}$, which is bounded below by $\Delta^r$ because the other factor is smaller than one. The remaining effects, $dr_1/d\sigma_2$, $dr_2/dt$ and $dr_2/d\sigma_1$, follow from the chain rule and have the sign of the corresponding output effect.
\end{proof}

\subsection{Comparative statics with respect to network externalities}\label{app:prod_cs_network}

\begin{appxprop}[Network externality effects under product R\&D]\label{propB:network_product}
Under the maintained assumptions of Section~\ref{app:prod_soc}, a stronger own-network externality increases own output and product R\&D, while a stronger rival's network externality reduces own output and R\&D when goods are substitutes ($m > 0$).
\end{appxprop}

\begin{proof}
Total differentiation with respect to $(b_1, b_2)$ yields:
\begin{align*}
\frac{dq_1}{db_1} &= \frac{1}{\Delta^r}\,\frac{1}{2(1-b_1)^2}\left[a + r_1 - \frac{c_1}{1-t} - mq_2\right]\!\left[1 - \frac{1}{2(1-b_2)(1-\sigma_2)\theta_2}\right] > 0, \\[4pt]
\frac{dq_1}{db_2} &= -\frac{m}{\Delta^r}\,\frac{a + r_2 - c_2 - mq_1}{4(1-b_1)(1-b_2)^2} \;\lessgtr\; 0 \;\text{ for }\; m \;\gtrless\; 0, \\[4pt]
\frac{dq_2}{db_2} &= \frac{1}{\Delta^r}\,\frac{1}{2(1-b_2)^2}[a + r_2 - c_2 - mq_1]\!\left[1 - \frac{1-t}{2(1-b_1)(1-\sigma_1)\theta_1}\right] > 0, \\[4pt]
\frac{dq_2}{db_1} &= -\frac{m}{\Delta^r}\,\frac{a + r_1 - \frac{c_1}{1-t} - mq_2}{4(1-b_1)^2(1-b_2)} \;\lessgtr\; 0 \;\text{ for }\; m \;\gtrless\; 0.
\end{align*}
The R\&D investment effects follow from the chain rule applied to $r_i(q_i)$:
\begin{align*}
\frac{dr_1}{db_1} &= \frac{1}{\Delta^r}\,\frac{1-t}{2(1-b_1)^2(1-\sigma_1)\theta_1}\left[a + r_1 - \frac{c_1}{1-t} - mq_2\right]\!\left[1 - \frac{1}{2(1-b_2)(1-\sigma_2)\theta_2}\right] > 0, \\[4pt]
\frac{dr_2}{db_2} &= \frac{1}{\Delta^r}\,\frac{1}{2(1-b_2)^2(1-\sigma_2)\theta_2}[a + r_2 - c_2 - mq_1]\!\left[1 - \frac{1-t}{2(1-b_1)(1-\sigma_1)\theta_1}\right] > 0.
\end{align*}
Cross effects ($dr_1/db_2$, $dr_2/db_1$) have the same sign as $-m$. The numerators are positive by the firms' first-order conditions, as in the proof of Proposition~\ref{propA:network_process}.
\end{proof}

\subsection{Optimal foreign subsidy}\label{app:prod_foreign}

The foreign government maximises $W_1 = \pi_1 - \sigma_1\frac{\theta_1 r_1^2}{2}$. Setting $dW_1/d\sigma_1 = 0$:
\[
-m(1-t)q_1\,\frac{dq_2}{d\sigma_1} - \sigma_1\,\theta_1\,r_1\,\frac{dr_1}{d\sigma_1} = 0.
\]
Substituting $dq_2/d\sigma_1$ and $dr_1/d\sigma_1$ from Section~\ref{app:prod_cs_policy} and simplifying yields:
\[
\sigma_1^* = \frac{\dfrac{m^2}{4(1-b_1)(1-b_2)}}{1 - \dfrac{1}{2(1-b_2)(1-\sigma_2)\theta_2}}.
\]
The slope of the reaction function is:
\[
\frac{d\sigma_1^*}{d\sigma_2} = \frac{1}{\left[1 - \frac{1}{2(1-b_2)(1-\sigma_2)\theta_2}\right]^2}\,\frac{m^2}{8(1-b_1)(1-b_2)^2(1-\sigma_2)^2\theta_2} > 0 \quad\text{for } m \neq 0.
\]
This establishes Proposition~\ref{prop:foreign_subsidy_product}.

\subsection{Optimal home subsidy and tax}\label{app:prod_home}

Substituting Eq.~\eqref{eq:govbudget} into Eq.~\eqref{eq:W2} and omitting the constant $L$, the home government maximises $W_2 = CS + \pi_2 - \sigma_2\frac{\theta_2 r_2^2}{2} + tp_1q_1$. The first-order condition for the subsidy simplifies to:
\[
\frac{dW_2}{d\sigma_2} = \left[q_1 + \frac{t}{1-t}\,c_1\right]\frac{dq_1}{d\sigma_2} + [q_2 + m(1-t)q_1]\,\frac{dq_2}{d\sigma_2} + tq_1\,\frac{dr_1}{d\sigma_2} - \frac{\sigma_2}{1-\sigma_2}\,q_2\,\frac{dr_2}{d\sigma_2} = 0.
\]
Setting $m = 0$ yields $\sigma_2^* = \frac{1}{1+2(1-b_2)}$, which establishes Proposition~\ref{prop:home_policy_product}.

The first-order condition for the optimal tax is:
\begin{multline*}
\frac{dW_2}{dt} = q_1\,\frac{d}{dt}(a_1 + r_1 - p_1) + q_2\,\frac{d}{dt}(a_2 + r_2 - p_2) + q_2\,\frac{dp_2}{dq_1}\,\frac{dq_1}{dt} \\
- \sigma_2\,\theta_2\,r_2\,\frac{dr_2}{dt} + p_1\,q_1 + t\left[p_1\,\frac{dq_1}{dt} + q_1\,\frac{dp_1}{dt}\right] = 0.
\end{multline*}
For $m \neq 0$, this equation together with the first-order condition for $\sigma_2$ forms a system that cannot be solved in closed form. The numerical solution is presented in Section~\ref{sec:simulations}.

%% file: supplementary_material.tex

\setcounter{section}{0}
\renewcommand{\thesection}{S.\arabic{section}}\renewcommand{\theHsection}{S.\arabic{section}}
\renewcommand{\thesubsection}{S.\arabic{section}.\arabic{subsection}}
\renewcommand{\thefigure}{S.\arabic{figure}}
\setcounter{figure}{0}
\renewcommand{\thetable}{S.\arabic{table}}
\setcounter{table}{0}
\renewcommand{\theequation}{S.\arabic{equation}}
\setcounter{equation}{0}

\begin{center}
{\Large\bfseries Supplementary Material}
\end{center}

\medskip

\noindent This supplement provides institutional background on the tax instrument (Section~\ref{sm:tax_instruments}), reports the optimal policy instruments as functions of network externality strength (Section~\ref{sm:policy_instruments}), presents numerical results for complementary goods (Section~\ref{sm:complements}), and examines the robustness of the results to the parameter values (Section~\ref{sm:robustness}). All parameter values are at the baseline specified in Section~\ref{sec:calibration} unless otherwise noted.

\section{Revenue-Based Taxation of Foreign Firms}\label{sm:tax_instruments}

The model's import tax $t$ is a proportional levy on the foreign firm's revenue from sales in the home market. Several policy instruments used in practice have this structure. They differ in legal form and scope, but each of them reduces the foreign firm's net receipts from its sales in the home market without directly changing the price consumers pay.

A recent example is the digital services tax (DST). Since 2019 around ten European OECD countries have adopted DSTs, which are levied on the gross revenue of large digital companies from specified digital activities in the taxing jurisdiction, at rates ranging from 1.5\% to 7.5\%, and fall mainly on US companies \citep{Enache_2025}. They were introduced as interim measures while the OECD's work on the tax challenges of the digital economy had not yet produced a multilateral consensus \citep{OECD_2018}, and they have led to trade tensions with the United States.\footnote{In February 2025 the US administration issued a presidential memorandum, \emph{Defending American Companies and Innovators From Overseas Extortion and Unfair Fines and Penalties}, directing the US Trade Representative to determine whether to renew its investigations of the DSTs of France, Austria, Italy, Spain, Turkey and the United Kingdom. The memorandum describes DSTs as ``designed to plunder American companies''.} India's Equalisation Levy had the same structure, at 6\% on online advertising from 2016 and 2\% on e-commerce supplies from 2020, both levied on non-resident firms' gross receipts and both since withdrawn. Withholding taxes on cross-border royalties, and increasingly on fees for technical services, are a longer-established instrument with the same structure. They are levied on the gross payment \citep{OECD_2018}, and the statutory rates on royalties lie between 10\% and 30\% in most jurisdictions and average 14.5\% \citep{OECD_2026}.

Table~\ref{tab:tax_summary} summarises these instruments. What they have in common is that they reduce the foreign firm's net return from its sales or activity in the home market. The parameter $t$ captures this common feature and abstracts from the legal differences between them.

\begin{table}[ht]
\centering
\caption{Revenue-based policy instruments and the model's tax parameter $t$}
\label{tab:tax_summary}
\renewcommand{\arraystretch}{1.25}
\begin{tabular}{@{}p{4cm} p{2.2cm} p{7cm}@{}}
\toprule
\textbf{Instrument} & \textbf{Typical rate} & \textbf{Mechanism} \\
\midrule
Digital services taxes & 1.5 to 7.5\% & Tax on large digital companies' revenue from specified digital activities in the home market \\[3pt]
Equalisation levies & 2 to 6\% & Tax on non-resident firms' gross receipts from e-commerce or advertising in the home market \\[3pt]
Withholding taxes on royalties and technical fees & 10 to 30\%\newline (royalties) & Proportional levy on cross-border payments for IP, licensing, and technical services \\
\bottomrule
\end{tabular}
\end{table}

\section{Policy Instruments as Functions of Network Externality Strength}\label{sm:policy_instruments}

This section reports the optimal policy instruments in the SPNE as functions of the common network externality parameter $b = b_1 = b_2$, for both R\&D specifications and both values of substitutability ($m = 0.01$ and $m = 0.25$). These figures complement the welfare-difference figures in the main text (Section~\ref{sec:simulations}) by showing how the policy mix itself changes as network externalities strengthen.

\subsection{Process R\&D}

Figure~\ref{fig:policy_process_b} plots the optimal foreign subsidy $s_1^*$, the optimal home subsidy $s_2^*$, and the optimal import tax $t^*$ as functions of $b$ under process R\&D.

\begin{figure}[ht]
    \centering
    \includegraphics[width=\textwidth]{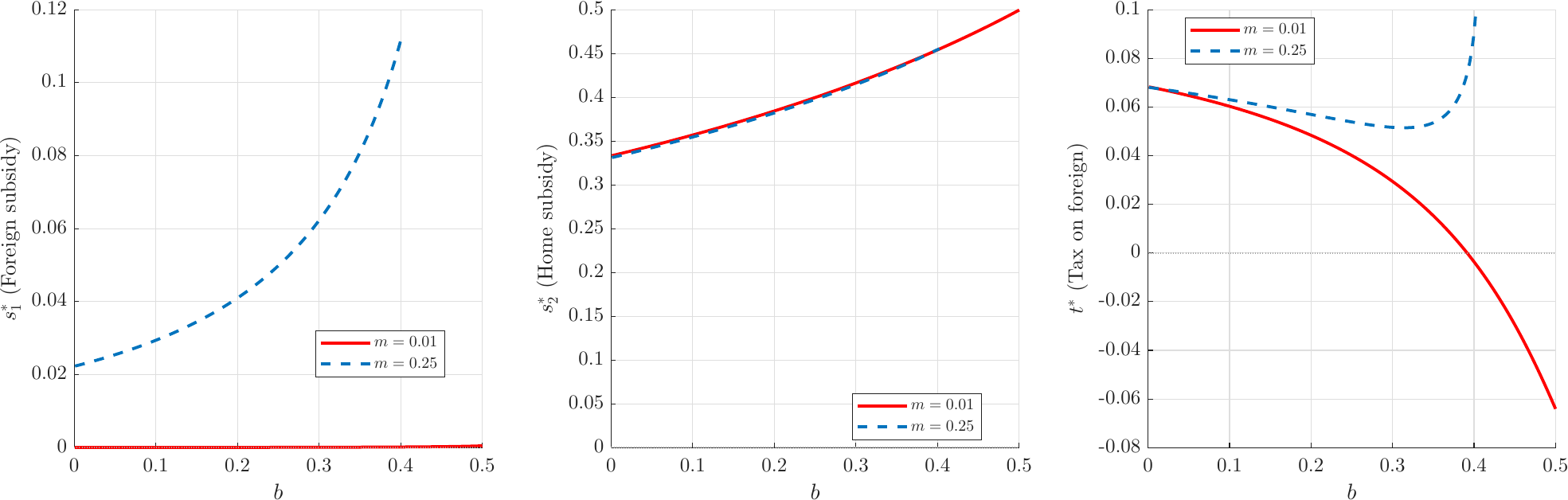}
    \caption{Process R\&D: optimal policy instruments as functions of network externality strength $b$, for $m = 0.01$ (solid) and $m = 0.25$ (dashed).}
    \label{fig:policy_process_b}
\end{figure}

The foreign subsidy $s_1^*$ increases in $b$, and by more at higher substitutability. The foreign government's incentive to subsidise increases as the network externality strengthens and as the goods become closer substitutes. At $m = 0.01$ the foreign subsidy is close to zero throughout, because the business-stealing motive that drives the foreign government to subsidise is weak when goods are nearly independent. The home subsidy $s_2^*$ also increases in $b$. This is consistent with the analytical result for independent goods, where the optimal subsidy is increasing in $b_2$ (Proposition~\ref{prop:home_policy_process}). The two curves for $s_2^*$ are nearly indistinguishable across the two values of $m$. This suggests that at these parameter values the home subsidy is driven primarily by network externality strength. The tax, by contrast, is sensitive to substitutability. At $m = 0.01$ the import tax $t^*$ declines in $b$, crosses zero at high $b$ and turns negative. At $m = 0.25$ it remains positive, declines over most of the range and increases again close to the upper boundary of the admissible region. As network externalities strengthen, the home government shifts the policy mix away from taxation and towards subsidisation.

\subsection{Product R\&D}

Figure~\ref{fig:policy_product_b} plots the corresponding instruments under product R\&D.

\begin{figure}[ht]
    \centering
    \includegraphics[width=\textwidth]{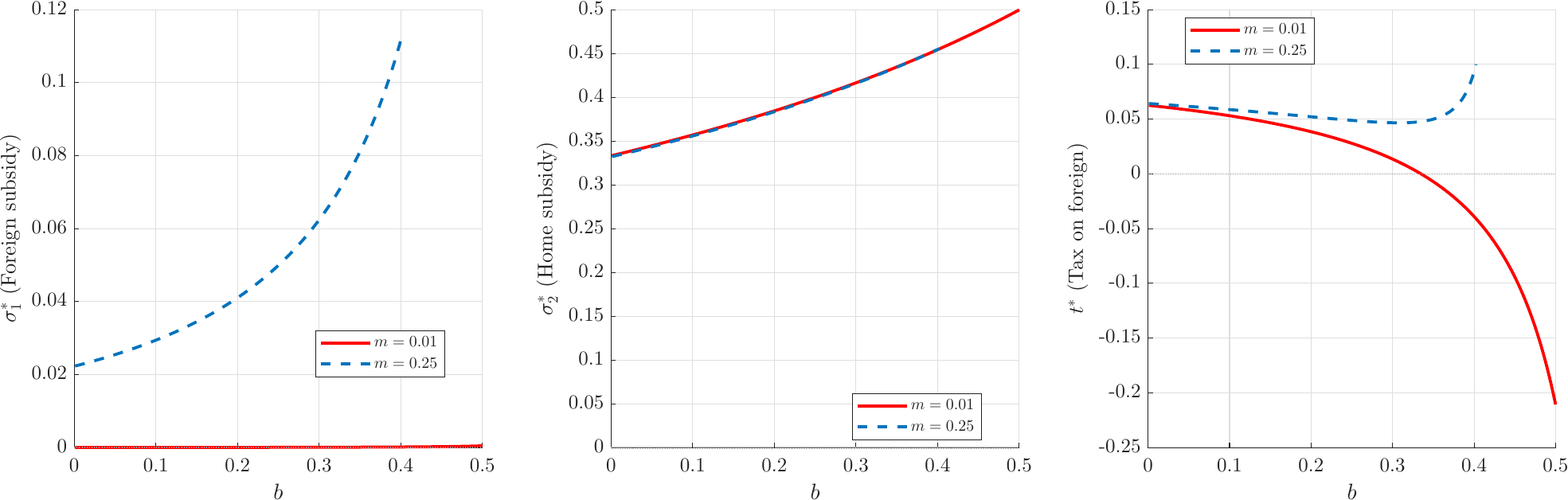}
    \caption{Product R\&D: optimal policy instruments as functions of network externality strength $b$, for $m = 0.01$ (solid) and $m = 0.25$ (dashed).}
    \label{fig:policy_product_b}
\end{figure}

The foreign subsidy $\sigma_1^*$ and the home subsidy $\sigma_2^*$ display qualitative patterns similar to those under process R\&D. The difference is in the tax. At $m = 0.01$, $t^*$ declines in $b$ and turns negative at a lower $b$ than under process R\&D. A negative tax is an import subsidy, and at high $b$ it is considerably larger than under process R\&D. This sign reversal is the mechanism behind the win-win result reported in Section~\ref{sec:sim_product}. When $t^* < 0$, the home government pays the foreign firm to expand its sales in the home market. The payment increases the foreign firm's return to investment in product R\&D, and the firm invests more. Home consumers' willingness to pay increases with quality, consumption expands, and at high $b$ the network externality increases willingness to pay again by enough that the gain in consumer surplus exceeds the fiscal cost of the subsidy. At $m = 0.25$, $t^*$ declines in $b$ over most of the range, increases close to the upper boundary, and does not cross zero within the admissible range. This is consistent with the absence of the win-win result at higher substitutability.

\section{Complementary Goods}\label{sm:complements}
 
Figures~\ref{fig:welfare_diff_process} and~\ref{fig:welfare_diff_product} in the main text illustrate the case of substitute goods ($m > 0$). This section summarises the analytical predictions and numerical results for the complementary-goods case ($m < 0$), which applies to settings where the two countries' products are used jointly rather than as alternatives. Under the intermediate-input interpretation of Remark~\ref{rem:isomorphism}, complementarity corresponds to sectors where the marginal product of each input in the shared downstream industry increases with the use of the other.

The best-response propositions in Sections~\ref{sec:processRD} and~\ref{sec:productRD} and the comparative statics in Appendices~\ref{app:derivations_process} and~\ref{app:derivations_product} cover both signs of $m$. When $m < 0$, three qualitative changes follow directly from the comparative statics.
 
One can see that the cross-policy effects reverse sign. Under substitutability, a subsidy to the domestic firm reduces the rival's output, and the import tax increases home output. Under complementarity, both effects are reversed. A subsidy to the domestic firm increases the rival's output, because higher domestic output increases demand for the complementary foreign good. The import tax reduces both foreign output and home output, because suppressing the foreign good also suppresses demand for the domestic good. The import tax harms rather than helps the home firm.
 
The foreign subsidy best response remains positive ($s_1^* > 0$ and $\sigma_1^* > 0$ for $m \neq 0$), and subsidies remain strategic complements ($ds_1^*/ds_2 > 0$). Both results follow from the quadratic dependence on $m$ in the best-response expressions (Propositions~\ref{prop:foreign_subsidy_process} and~\ref{prop:foreign_subsidy_product}). The foreign government subsidises its firm regardless of whether goods are substitutes or complements, though the magnitude and welfare consequences differ.
 
Importantly, the home government's optimal policy mix shifts towards heavier reliance on the subsidy. The tax becomes less attractive, because the revenue it collects is offset by the welfare cost of reducing foreign output, on which domestic consumers depend. For complementary goods, the home government may find it optimal to set a negative import tax (an import subsidy) at lower values of $b$ than under substitutes. This is because an import subsidy expands foreign output, and under complementarity larger foreign output increases demand for the domestic good. The network externality works in the same direction. Larger consumption of the foreign good increases the willingness to pay for it (the direct externality), and for complements the higher willingness to pay raises consumption of the foreign good further, which increases the demand for the domestic good (the indirect cross-demand externality of Remark~\ref{rem:indirect}). The network externality therefore augments the business-enhancing effect. This contrasts with the substitute case, where the tax is positive at low $b$ and turns negative only at high $b$.

\subsection{Numerical results}

Figures~\ref{fig:welfare_diff_process_comp} and~\ref{fig:welfare_diff_product_comp} report the welfare differences between the SPNE and laissez-faire for $m = -0.10$ (solid) and $m = -0.25$ (dashed) under process and product R\&D respectively. Figures~\ref{fig:policy_process_comp} and~\ref{fig:policy_product_comp} report the corresponding policy instruments.

\begin{figure}[ht]
    \centering
    \includegraphics[width=\textwidth]{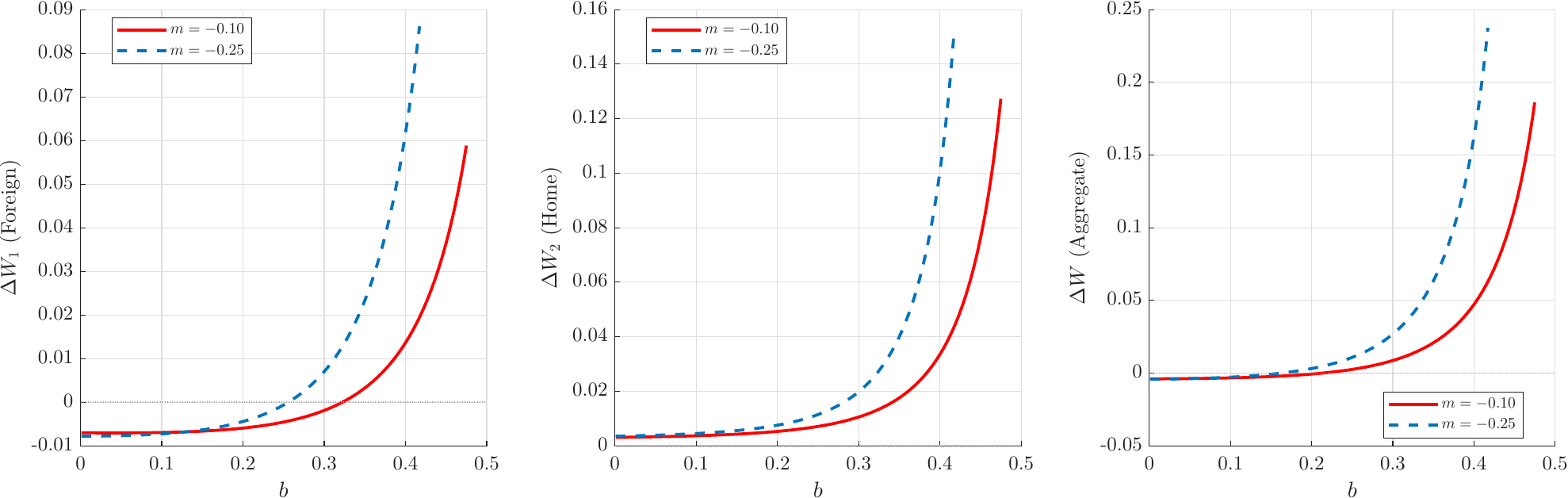}
    \caption{Process R\&D, complements: welfare differences between the SPNE and laissez-faire as functions of network externality strength $b$, for $m = -0.10$ (solid) and $m = -0.25$ (dashed). Left: foreign country switches from losing to gaining at high $b$. Centre: home country gains throughout. Right: aggregate welfare turns positive and the gain is substantial at high $b$.}
    \label{fig:welfare_diff_process_comp}
\end{figure}

\begin{figure}[ht]
    \centering
    \includegraphics[width=\textwidth]{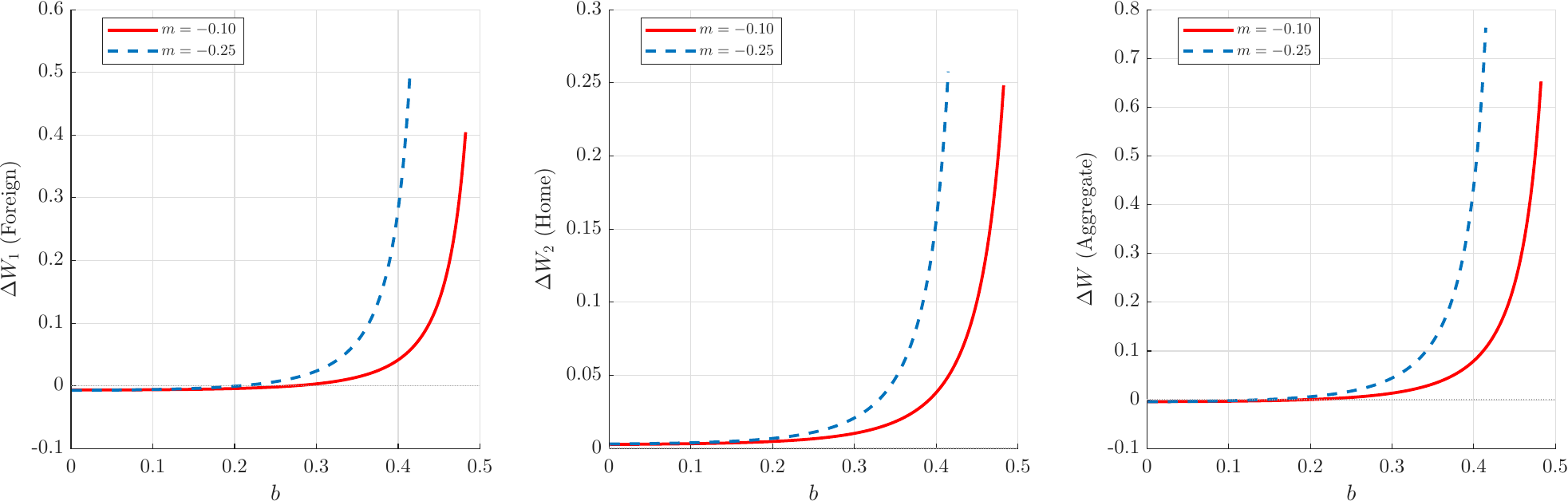}
    \caption{Product R\&D, complements: welfare differences between the SPNE and laissez-faire as functions of network externality strength $b$, for $m = -0.10$ (solid) and $m = -0.25$ (dashed). Left: foreign country switches from losing to gaining at high $b$. Centre: home country gains throughout. Right: aggregate welfare gains are substantially larger than under process R\&D.}
    \label{fig:welfare_diff_product_comp}
\end{figure}

\begin{figure}[ht]
    \centering
    \includegraphics[width=\textwidth]{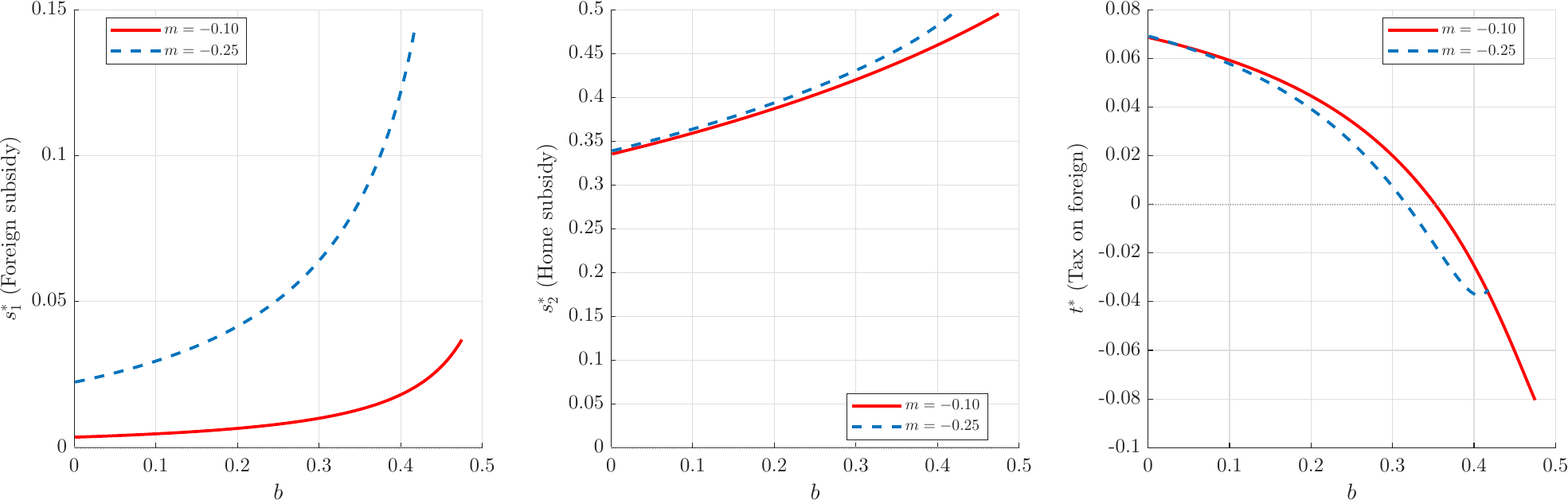}
    \caption{Process R\&D, complements: optimal policy instruments as functions of network externality strength $b$, for $m = -0.10$ (solid) and $m = -0.25$ (dashed).}
    \label{fig:policy_process_comp}
\end{figure}

\begin{figure}[ht]
    \centering
    \includegraphics[width=\textwidth]{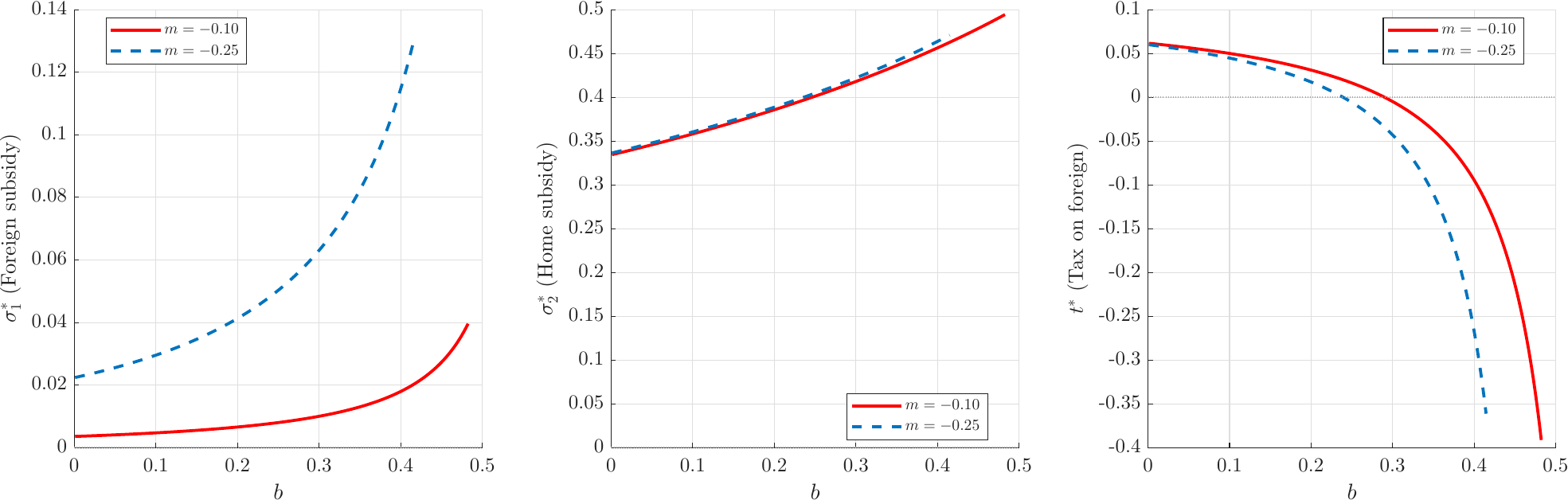}
    \caption{Product R\&D, complements: optimal policy instruments as functions of network externality strength $b$, for $m = -0.10$ (solid) and $m = -0.25$ (dashed).}
    \label{fig:policy_product_comp}
\end{figure}

The complements case clarifies the role of the business-stealing effect. When goods are complements, the business-stealing cost that prevents the win-win outcome for closer substitutes reverses sign. Subsidising the foreign firm increases rather than reduces demand for the domestic good.

The win-win outcome appears under both R\&D specifications. Under process R\&D with complements, the home government keeps the tax low (Figure~\ref{fig:policy_process_comp}), because taxing the foreign firm would also suppress demand for the domestic good; the tax is positive where the mutual gains begin and turns negative only when the externality is sufficiently strong. The home subsidy increases domestic output, which increases demand for the foreign good because the goods are complements. The foreign subsidy increases foreign output, which increases demand for the domestic good in the same way. Each good's network externality adds to this expansion, since for complements the indirect cross-demand externality works in the same direction as the direct one. Both countries can gain without an import subsidy. Under product R\&D, the home government goes further and sets a deeply negative tax at high $b$ (Figure~\ref{fig:policy_product_comp}). The consumer-surplus gain from the foreign good's quality improvement is larger under complementarity, because higher foreign quality increases willingness to pay for both goods.

Furthermore, the welfare gains under complements are substantially larger than under substitutes. Under product R\&D the aggregate welfare difference reaches values several times larger than in the substitutes case (compare Figure~\ref{fig:welfare_diff_product_comp} with Figure~\ref{fig:welfare_diff_product}). Stronger complementarity ($m = -0.25$) produces an earlier and larger win-win outcome than weaker complementarity ($m = -0.10$). This is the opposite of the substitutes case, in which stronger substitutability ($m = 0.25$) eliminates it.

The complements case connects directly to the supply-chain isomorphism of Remark~\ref{rem:isomorphism}. When the foreign and domestic goods are used together in a downstream industry, the case for subsidising the foreign firm is strongest. The CHIPS Act's incentives for foreign semiconductor manufacturers are an example. The foreign firm's chips are complements to domestically produced equipment and design tools, not substitutes.

\section{Robustness to Parameter Values}\label{sm:robustness}

The numerical exercises in the main text fix the demand intercept, the marginal cost and the R\&D cost parameters. The equilibrium conditions are homogeneous of degree one in $(a, c)$. If $a$ and $c$ are multiplied by the same factor, quantities, prices and R\&D levels are multiplied by that factor, welfare net of the constant $L$ by its square, and the tax and subsidy rates and the signs of all welfare differences are unchanged. Hence $a = 1$ is a normalisation, and with symmetric firms each R\&D specification has two remaining parameters, the cost ratio $c/a$ and the R\&D cost parameter ($\varphi$ under process R\&D, $\theta$ under product R\&D).

Figures~\ref{fig:robust_c} and~\ref{fig:robust_rd} show the foreign country's break-even locus ($\Delta W_1 = 0$) in the $(m,b)$ space of Figure~\ref{fig:winwin_map} for $c/a \in \{0.5, 0.6, 0.7, 0.8, 0.9\}$ and for R\&D cost parameters in $\{2, 2.5, 3, 4\}$. The thick solid line is the parameterisation of the main text. The home country gains at every admissible point in all cases, so the win-win set lies above each locus.

\begin{figure}[ht]
    \centering
    \begin{subfigure}[b]{0.49\textwidth}
        \includegraphics[width=\textwidth]{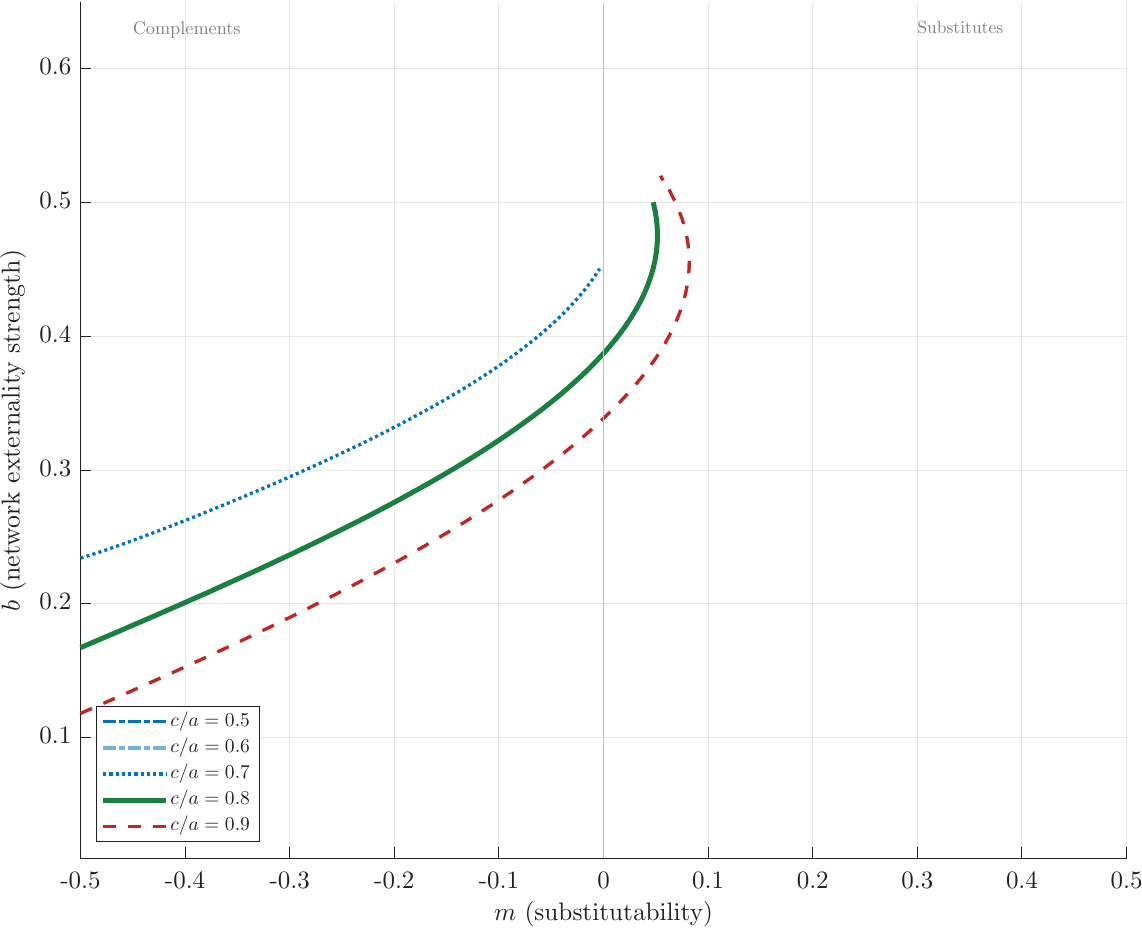}
        \caption{Process R\&D}
    \end{subfigure}
    \hfill
    \begin{subfigure}[b]{0.49\textwidth}
        \includegraphics[width=\textwidth]{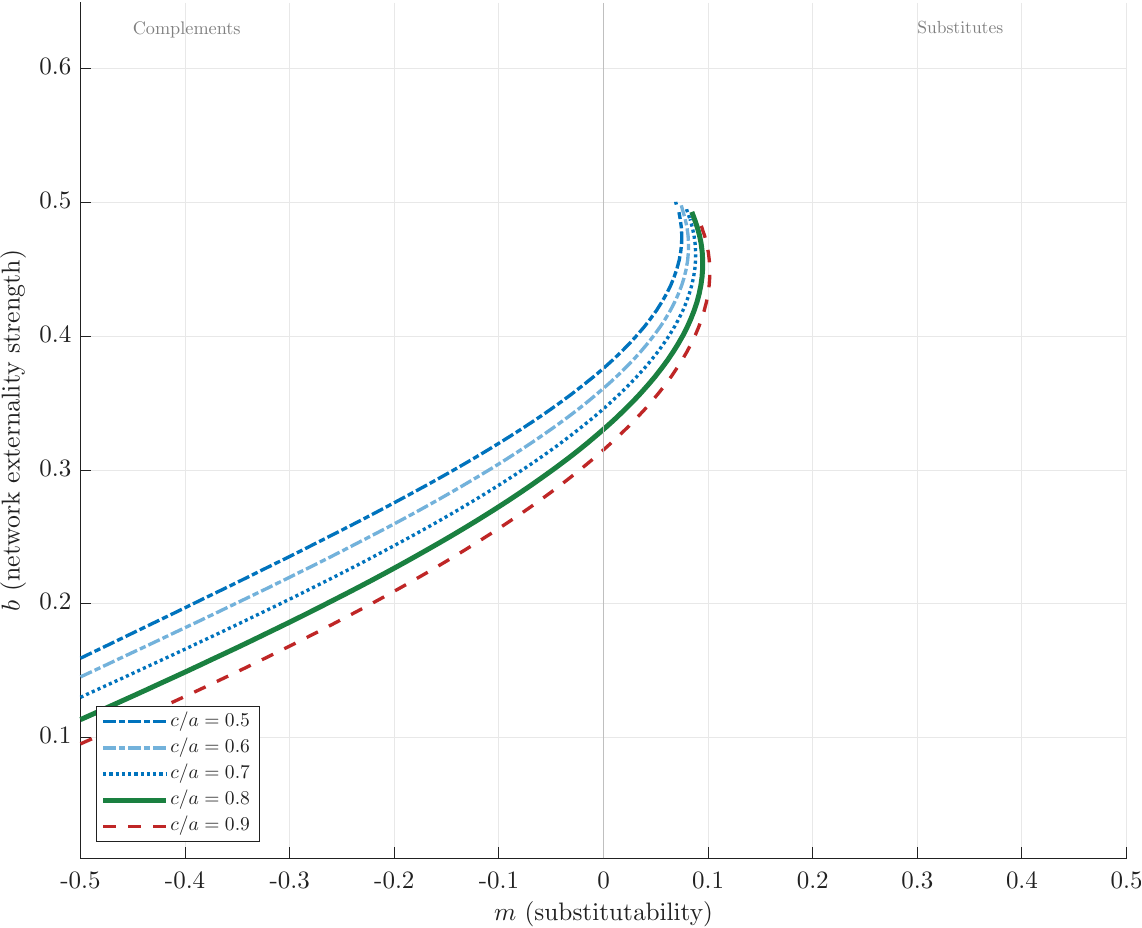}
        \caption{Product R\&D}
    \end{subfigure}
    \caption{The foreign country's break-even locus ($\Delta W_1 = 0$) for different cost ratios $c/a$, with the R\&D cost parameter equal to $2.5$. No locus appears where the win-win set is empty.}
    \label{fig:robust_c}
\end{figure}

\begin{figure}[ht]
    \centering
    \begin{subfigure}[b]{0.49\textwidth}
        \includegraphics[width=\textwidth]{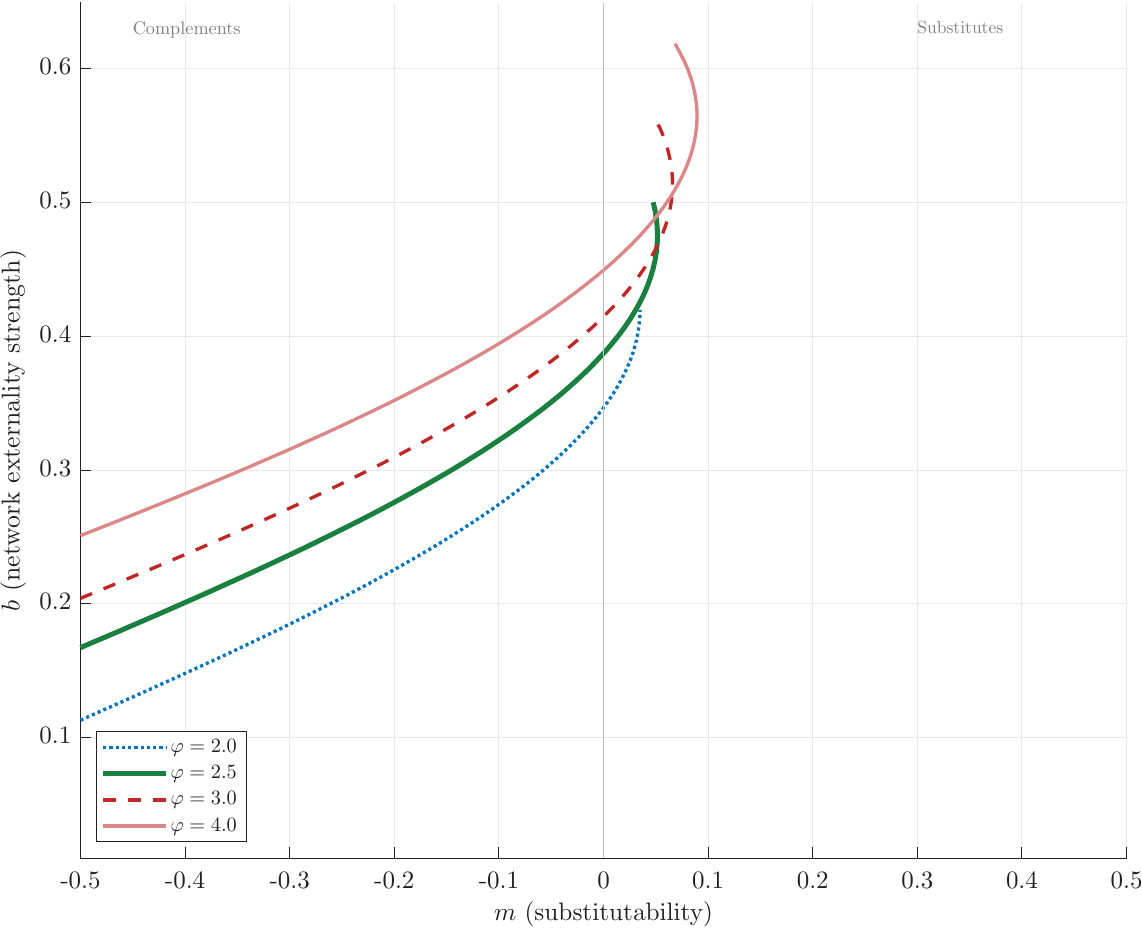}
        \caption{Process R\&D}
    \end{subfigure}
    \hfill
    \begin{subfigure}[b]{0.49\textwidth}
        \includegraphics[width=\textwidth]{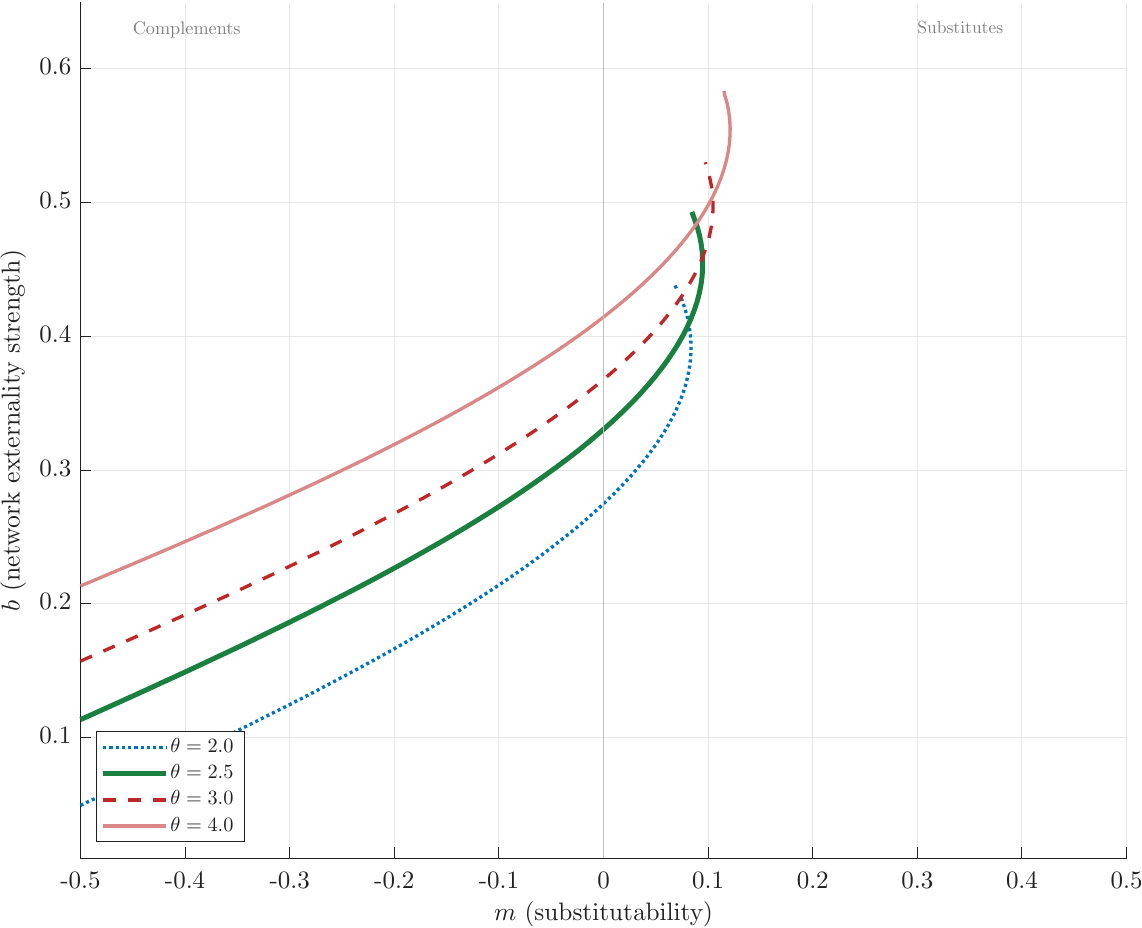}
        \caption{Product R\&D}
    \end{subfigure}
    \caption{The foreign country's break-even locus ($\Delta W_1 = 0$) for different R\&D cost parameters, with $c/a = 0.8$.}
    \label{fig:robust_rd}
\end{figure}

One can see that the qualitative patterns are the same for all parameter values considered. The threshold for $b$ increases in $m$; when the goods are substitutes, the optimal import tax is negative at every point of the win-win set; and the threshold under process R\&D is never below the threshold under product R\&D. Where both specifications deliver mutual gains, the aggregate and the foreign welfare gains are larger under product R\&D. Under product R\&D the win-win set exists for all parameter values considered and includes weak substitutes.

Under process R\&D the win-win set is more sensitive to the parameter values. It is empty for low cost ratios, contains only complements for intermediate cost ratios, and extends to weak substitutes for higher cost ratios, as in the parameterisation of the main text; it is also larger for more costly R\&D. In these cases the optimal import tax under process R\&D also turns negative when the externality is sufficiently strong.